\documentclass[11pt]{article}

\PassOptionsToPackage{final}{hyperref}

\usepackage[margin=1in]{geometry}

\usepackage{booktabs}

\usepackage{makecell}
\usepackage{multirow}

\usepackage{float}
\usepackage{comment}
\usepackage{algorithm}
\usepackage[noend]{algpseudocode}

\usepackage{amsmath,amssymb,mathtools}
\usepackage{amsthm}
\usepackage{thmtools} %

\usepackage{xcolor}

\usepackage{soul}
\usepackage[normalem]{ulem}
\usepackage{graphicx}
\usepackage{dsfont}

\usepackage{tikz}
\usetikzlibrary{calc}
\usetikzlibrary{shapes.geometric}
\usetikzlibrary{automata, positioning, calc, fit, shapes.misc}
\usetikzlibrary{matrix}
\usetikzlibrary{positioning}
\usetikzlibrary{arrows, decorations.pathmorphing}
\usetikzlibrary{decorations.pathreplacing,calligraphy}
\usepackage{paralist}

\usepackage{caption}
\usepackage[labelformat=simple]{subcaption}

\usepackage[multiple]{footmisc}
\usepackage{aliascnt}

\usepackage{rotating}
\usepackage{tcolorbox}

\usepackage{csquotes}
\usepackage{colortbl}

\usepackage{sansmath}

\newcolumntype{C}[1]{>{\centering\arraybackslash}p{#1}}

\usepackage[backref=page]{hyperref}
\hypersetup{colorlinks=true,citecolor=blue,linkcolor=blue}
\hypersetup{colorlinks=true}
\hypersetup{linkcolor=[rgb]{.7,0,0}}
\hypersetup{citecolor=[rgb]{0,.7,0}}
\hypersetup{urlcolor=[rgb]{.7,0,.7}}

\declaretheorem[numberwithin=section]{theorem}
\declaretheorem[sibling=theorem, style=definition]{definition}
\declaretheorem[sibling=theorem]{lemma}

\declaretheorem[sibling=theorem]{observation}

\declaretheorem[sibling=theorem]{proposition}

\newcommand{\addQEDstyle}[2]{\AtBeginEnvironment{#1}{\pushQED{\qed}\renewcommand{\qedsymbol}{#2}}\AtEndEnvironment{#1}{\popQED}} 
\addQEDstyle{definition}{$\triangle$}

\makeatletter
\patchcmd{\ALG@step}{\addtocounter{ALG@line}{1}}{\refstepcounter{ALG@line}}{}{}
\newcommand{\ALG@lineautorefname}{Line}
\makeatother

\makeatletter
\newcommand\Ps@textstyle[2]{\mathbb{P}_{#1}\left[{#2}\right]}
\newcommand\Es@textstyle[2]{\mathbb{E}_{#1}\left[{#2}\right]}
\newcommand\Ps[2]{%
  \mathchoice %
  {\underset{{#1}}{\mathbb{P}}\left[{#2}\right]}
  {\Ps@textstyle{#1}{#2}}
  {\Ps@textstyle{#1}{#2}}
  {\Ps@textstyle{#1}{#2}}
}
\newcommand\Es[2]{%
  \mathchoice %
  {\underset{{#1}}{\mathbb{E}}\left[{#2}\right]}
  {\Es@textstyle{#1}{#2}}{\Es@textstyle{#1}{#2}}{\Es@textstyle{#1}{#2}}
}
\makeatother

\makeatletter
\DeclareRobustCommand
  \myvdots{\vbox{\baselineskip4\p@ \lineskiplimit\z@
    \kern4\p@ \hbox{.}\hbox{.}\hbox{.}}}
\makeatother

\usepackage{framed}

\newcommand{\Appls}{\mathcal{A}}
\newcommand{\Insts}{\mathcal{I}}

\DeclareMathOperator*{\vio}{vio}

\newcommand{\DA}{\mathsf{DA}}
\newcommand{\APDA}{\mathsf{APDA}}
\newcommand{\IPDA}{\mathsf{IPDA}}
\newcommand{\PN}{\mathsf{PN}}
\newcommand{\Stab}{\mathsf{Stab}}

\newcommand{\Legal}{\mathsf{L}}
\newcommand{\K}{\mathsf{K}}
\newcommand{\EADAM}{\mathsf{EADAM}}

\newcommand{\match}{\mathsf{match}}
\newcommand{\rots}{\mathsf{rots}}

\newcommand{\covers}{\gtrdot}

\begin{document}

\title{
Characterization of Priority-Neutral Matching Lattices\thanks{We thank Ata Atay, Nicole Immorlica, Phil Reny, participants at Stony Brook 2024, and various anonymous reviewers for helpful comments and conversations.}
}

\hypersetup{pdftitle={Characterization of Priority-Neutral Matching Lattices}}
\hypersetup{pdfauthor={Clayton Thomas}}

\author{
Clayton Thomas\thanks{Yale University. 
| \emph{E-mail}: \href{mailto:thomas.clay95@gmail.com}{thomas.clay95@gmail.com}.
}
}

\date{\today}

\begin{titlepage}
\maketitle
\begin{abstract}
  We study the structure of the set of priority-neutral matchings.
These matchings, introduced by \cite{Reny22}, generalize stable matchings by allowing for priority violations in a principled way that enables Pareto-improvements to stable matchings. 
Known results show that the set of priority-neutral matchings is a lattice, suggesting that these matchings may enjoy the same tractable theoretical structure as stable matchings.

In this paper, we characterize priority-neutral matching lattices, 
and show that their structure is considerably more intricate than that of stable matching lattices. 
To begin, we show priority-neutral lattices are not distributive, an important property that characterizes stable lattices and is satisfied by many other lattice structures considered in matching theory and algorithm design.
Then, in our main result, we show that priority-neutral lattices are in fact characterized by a more-involved property which we term being a ``movement lattice,'' which allows for significant departures from the order theoretic properties of distributive (and hence stable) lattices.
While our results show that priority-neutrality 
is more intricate than stability, they also establish tractable properties.
Indeed, as a corollary of our main result, we obtain the first known polynomial-time algorithm for checking whether a given matching is priority-neutral.

\end{abstract}
\thispagestyle{empty}
\end{titlepage}

\tableofcontents
\thispagestyle{empty}
\clearpage
\pagenumbering{arabic}

\section{Introduction}
\label{sec:intro}

Stable matchings are a foundational notion in the theory of matching under preferences.
Since their introduction in \cite{GaleS62}, they have been found to exhibit a both textbook algorithmic elegance and rich theoretical structure \cite{Knuth76, blair1984every, GusfieldI89, kleinberg2006algorithm}.
These insights culminate in a striking characterization: when matchings are ordered according to agents' preferences, the set of stable matchings forms a \emph{distributive lattice}.\footnote{
  Specifically, matchings are partially ordered according to preferences as follows: if we are matching students to schools, $\mu \ge \nu$ means that for all students $d$, we have $\mu(d)\succeq_d\nu(d)$.
  A lattice is a partially-ordered set such that each pair $x,y$ has a greatest lower bound $x \wedge y$ and a least upper bound $x \vee y$.
  A distributive lattice is one in which the equation $\mu \wedge (\nu \vee \xi) = (\mu \wedge \nu) \vee (\mu \wedge \xi)$ always holds. 
}
Moreover, every distributive lattice arises in this way, so this property precisely characterizes the structure of stable matchings.
This characterization has been instrumental for countless results on stable matchings, including their combinatorics (e.g., \cite{KarlinGW18,PalmerP21}), communication (e.g., \cite{GonczarowskiNOR19, AshlagiBKS20, GonczarowskiT24}), and algorithms (e.g., \cite{ManloveMatchPrefs13, echenique2023online}).

Despite these successes, there are compelling reasons to look beyond stable matchings.
For instance, in the crucial real-world setting of student-to-school assignment \cite{AbdulkadirougluS03,abdulkadirouglu2005boston,abdulkadirouglu2005new, sonmez2023minimalist},
stability may not be a binding requirement, and is often overly restrictive.
First, regardless of the schools' priorities, students are typically unable to attend a school that the centralized algorithm did not assign them to.
Thus, priority violations in matching $\mu$---i.e., pairs $(d,h)$ where $h\succ_d \mu(d)$ and $d \succ_h \mu(h)$---may create a sense of unfairness, but may not compromise the integrity of the matching overall.
Second, stability implies forgoing Pareto-optimality for the students.
For instance, \cite{AbdulkadirougluPR09} report that in 2006 in New York City, over 4,000 grade 8 students could be made better off without worsening the match of any student.
This leads to an important question: is there a natural class of matchings which permits Pareto-improvements to stable matchings, while paying due respect to priority violations?\footnote{
  The classical Pareto-optimal matching mechanism in student-to-school assignment is Top Trading Cycles (TTC) \cite{ShapleyS74,AbdulkadirougluS03}.
  Despite having many desirable properties, TTC (informally speaking) ignores priority violations completely, motivating the search for a more integrated compromise between stability and Pareto-optimality.
}

Motivated by this question, the recent breakthrough work \cite{Reny22} proposes a generalization of stability, termed {priority-neutrality}.
When matching applicants to institutions,
a matching $\mu$ is \emph{priority-neutral} if, for each applicant $d$ whose priority is violated by $\mu$, it is not possible for another matching $\nu$ to make $d$ better off, unless $\nu$ also violates the priority of some applicant who is made worse off.
Since stable matchings violate no applicant's priority, every stable matching is priority-neutral. 
But unlike stable matchings, priority-neutral matchings also allow for violations of $d$'s priority when, in some sense, $d$ cannot be improved without undue effects on others.

\cite{Reny22} proves that priority-neutral matchings always Pareto-improve upon stable matchings, and include a Pareto-optimal matching. 
Moreover, he shows their set has structure: like stable matchings, the set of priority-neutral matchings forms a lattice.\footnote{
  \cite{Reny22} is far from the first work to pursue the agenda of Pareto-improving stable matchings; see related work below. 
  Perhaps most notably, \cite{EhlersM20} introduce \emph{legal} matchings, where a set $\Legal$ of matchings is called legal if $\Legal = \{ \mu\ \allowbreak|\ \allowbreak \text{$\forall \nu \in \Legal:$ \allowbreak no applicant $d$'s priority is violated by $\mu$ at institution $\nu(d)$} \}$.
  \cite{EhlersM20} prove there is a unique legal set $\Legal$, which always contains all stable matchings while allowing for Pareto-improvements.
  Moreover, $\Legal$ is a lattice.
  However, the self-referential nature of the definition of a legal set motivates the search for other notions such as priority-neutrality, which \cite{Reny22} posits may be a simpler notion than legality.
  \label{footnote:legal}
}
Thus, known results seem to suggest that---in addition to offering improvements in applicants' welfare---priority-neutral matchings may be as theoretically tractable as stable matchings.

\paragraph{Our Results.}

In this paper, we show that the structure of the priority-neutral lattice is in fact \emph{far} more involved than that of the stable lattice.
Our first result (\autoref{sec:non-distributive}) begins by answering an open question of \cite{Reny22}, and showing that the priority-neutral lattice is \emph{not} distributive.
Moreover, we show that in the priority-neutral lattice, the greatest lower bound of two matchings is \emph{not} their coordinatewise minimum (as it is in the stable matching lattice, and in many other lattices of algorithmic interest, e.g., \cite{PicardQ80, KhullerNK93, GulS99}).\footnote{
  The coordinatewise minimum of matchings $\mu$ and $\nu$ pairs each applicant $d$ to their worse choice from $\{\mu(d), \nu(d)\}$.
  In the stable matching lattice, $\mu\wedge\nu$ is the coordinatewise minimum of $\mu$ and $\nu$, and $\mu\vee\nu$ is (the analogous notion of) the coordinatewise maximum of $\mu$ and $\nu$.
  One can show that, if $\wedge$ and $\vee$ are given by coordinatewise minima and maxima (respectively), then the lattice is distributive.
  We prove that in the priority-neutral lattice, $\wedge$ is \emph{not} given by coordinatewise minima (although, as shown by \cite{Reny22}, $\vee$ is given by coordinatewise maxima).
}
These facts imply that priority-neutral matchings lack much of the tractable structure of stable matchings.
(For instance, they show that the priority-neutral lattice cannot be represented by a partial order on a set of rotations, the classic strengthening of Birkhoff's representation theorem \cite{Birkhoff37} for the stable lattice which underpins the results discussed above.)

Given that priority-neutral lattices thus lack the tractable structure of stable matching lattices, we ask the following questions:
What lattices \emph{do} arise as priority-neutral matching lattices?
What properties to priority-neutral lattices satisfy, and how can we represent their structure?

In our main result (\autoref{sec:result}), we characterize priority-neutral matching lattices.
That is, we provide a representation of the set of priority-neutral matchings,
and moreover, we delineate the exact class of lattices which can arise in this way.
This class---which we call \emph{movement lattices}---is fairly intricate, and we show via examples that priority-neutral lattices can have order-theoretic structure quite different from that of distributive lattices. 
Moreover, proving that this representation captures priority-neutral matching lattices is involved: we build on an extensive theory of different matching lattices that generalize stability, and use techniques from \cite{kesten2010school,TangY14,EhlersM20,Reny22,FaenzaZ22}.
For some preliminary intuition, see below.

While our results imply that priority-neutral lattices are more intricate than stable ones, they also show ways in which priority-neutrality is tractable.
For instance, our characterization directly yields a poly-time algorithm for checking whether a given matching is priority-neutral.
Prior to our work, it was far from clear how to efficiently check whether a matching $\mu$ is priority-neutral,
since the definition of priority-neutrality requires verifying that no alternative matching $\nu$ satisfies some condition (specifically, $\nu$ cannot improve the match of an applicant whose priority is violated in $\mu$ without violating the priority of an applicant who is worse off in $\nu$).

\paragraph{Intuition for Our Results.}

Intuitively, the structure of the stable matching lattice is derived from the fact that avoiding priority violations---i.e., pairs $(d,h)$ with $h\succ_d \mu(d)$ and $d \succ_h \mu(h)$---enforces simple ordering-based constraints on the set of stable matchings.
For instance, in any stable matching $\mu$ which sets $\mu(h)=d_0$, consider any potential priority-violating pair $(d,h)$ with $d \succ_h d_0$.
Then, to preserve stability, we must guarantee that $\mu$ matches $d$ at least as well as $h$, i.e., $\mu(d) \succ_d h$.
Thus, this potential priority violation gives a natural ordering structure on the \emph{differences} between stable matchings: any time $h$ is made worse-off than $d$, we must in turn make $d$ better-off than $h$.

For a simple example of the above discussion, consider \autoref{fig:simple-ex}.
In this instance, we cannot assign $h_D$ to $d_C$ without ``first'' assigning $d_A$ to $h_B$ (or else $d_A$'s priority would be violated at $h_D$).
More generally, all stable matching lattices can be represented this way, where different changes can be made to the matching so long as \emph{all} of their ``predecessor'' changes have already been made. 
(The idea of such predecessor relations exactly characterizes distributive lattices by Birkhoff's representation theorem \cite{Birkhoff37}, and the formal notion capturing this for stable matchings is the rotation DAG; see \autoref{sec:aditional-prelims-rotation-DAG}.\footnote{
  In \autoref{sec:exposition-stable-construction}, we show for the sake of exposition that \autoref{fig:simple-ex} can be directly generalized to construct a stable matching lattice isomorphic to any distributive lattice.
  We believe our construction simplifies prior ones showing this fact \cite{blair1984every, GusfieldI89}.
})

\begin{figure}[htb]
\begin{minipage}{0.45\textwidth}
  \centering
  Institutions' Priorities:
  \\[0.05in]
  \begin{tabular}{ccccccc}
      \toprule
      $h_A$ & $h_B$ & $h_C$ & $h_D$
      \\ \midrule
      $d_A$ & $d_B$ & $d_C$ & $d_D$
      \\
      $d_B$ & $d_A$ & $d_D$ & $d_A$
      \\
      &  & 
      & $d_C$ & 
  \end{tabular}
\end{minipage} 
\begin{minipage}{0.45\textwidth}
  \centering
  Applicants' Preferences:
  \\[0.05in]
  \begin{tabular}{ccccccc}
      \toprule
      $d_A$ & $d_B$ & $d_C$ & $d_D$
      \\ \midrule
      $h_B$ & $h_A$ & $h_D$ & $h_C$
      \\
      $h_D$ & $h_B$ & $h_C$ & $h_D$
      \\
      $h_A$ & & & & 
  \end{tabular}
\end{minipage} 
\\ %
\begin{minipage}{\textwidth}
  \centering
  \begin{align*}
    \mu_1 & = \{ (h_A, d_A), (h_B, d_B), (h_C, d_C), (h_D, d_D)\} \\
    \mu_2 & = \{ (h_A, d_B), (h_B, d_A), (h_C, d_C), (h_D, d_D)\} \\
    \mu_3 & = \{ (h_A, d_B), (h_B, d_A), (h_C, d_D), (h_D, d_C)\}
  \end{align*}
\end{minipage} 
\caption{An example instance where the stable matching lattice is $\{ \mu_1, \mu_2, \mu_3\}$, with $\mu_1 < \mu_2 < \mu_3$.
This illustrates the simple ordering constraints on stable matchings: we must ensure that $d_A$ is improved ``before'' $h_D$ is made worse.
In contrast, we show that the ordering constraints on priority-neutral matchings are more involved: 
they may require only that \emph{some} change, among many possibilities, is made ``before'' another change.}
\label{fig:simple-ex}
\end{figure}

Where stable matchings are those that do not have priority violations, priority-neutral matchings are those matchings $\mu$ that do not have what we call \emph{priority-correcting adjustments}---i.e., alternative matchings $\nu$ that have $\nu(d) \succ_d \mu(d)$ for some $d$ whose priority is violated in $\mu$, and such that there does not exist a $d_x$ whose priority is violated in $\nu$ and $\mu(d)\succ_{d_x} \nu(d)$.
According to \cite{Reny22}, such matchings $\nu$ provide a plausible argument for such a applicant $d$ as to why their priority should not have been violated in matching $\mu$, namely, $\nu$ ``corrects'' $\mu$ from $d$'s perspective by improving applicant $d$'s match, while only harming applicants $d_x$ if their priorities are respected.

Unfortunately, avoiding priority-correcting adjustments \emph{does not} enforce the same simple ordering structure on the set of priority-neutral matchings (compared to the simple structure that avoiding priority violations enforces on the set of stable matchings).
Intuitively, this is because if $\mu$ violates $d$'s priority, and if a potential priority-correcting adjustment $\nu$ violates the priority of multiple applicants $d_x$ and $d_y$, then this $\nu$ can witness the fact that $\mu$ is not priority-neutral only if \emph{neither} $d_x$ or $d_y$ have been made better off than their match in $\nu$.
Hence, the change that violates $d$'s priority can be made after \emph{either} $d_x$ or $d_y$ have been improved. 
The resulting ordering structure can be much more intricate than those that (by Birkhoff's theorem) underpin stable and distributive lattices (which could only require, for example, that \emph{both} $d_x$ and $d_y$ must have been improved).

In our first result (in \autoref{sec:non-distributive}), we construct an instance implementing the above idea for priority-neutral lattices, showing that they are not distributive.
After this (in \autoref{sec:building}), we directly pursue a generalization of the above discussion, and attempt to capture the order-theoretic constraints induced by priority-correcting adjustments in terms of a ``multiple predecessor'' generalization of results for stable lattices.
While this approach allows us to prove some structural results on the set of priority neutral matchings, it does not give a tight characterization (see \autoref{sec:remaining-challenges} for discussion). 
Hence, in our main result (in \autoref{sec:result}), we take a more ``low level'' approach, and encode the data of a priority-correcting adjustment into a direct representation.
In the end, this low-level approach gives a tight characterization.

\paragraph{Related Work.}

Our paper builds most directly on two strands of literature.
First, we build on a vast literature on the structure of stable matchings following from the seminal \cite{GaleS62}.
For structural properties specifically, \cite{GusfieldI89} is the classical reference; see also \cite[Section 1.4]{echenique2023online} or \cite{CaiT19}.
Key individual works include \cite{Knuth74, blair1984every, IrvingL86}, which establish that stable matching lattices are exactly characterized by distributive lattices.
Other works establishing fundamental properties of stable matching mechanisms include, for example, \cite{GaleMsMachiavelli85, McVitieW71, Roth82-DA, Roth84, DubinsMachiavelliGaleShapley81, GaleS85}.

Second, we build on papers studying Pareto-improvements of stable matchings.
\cite{kesten2010school} introduces a mechanism known as $\EADAM$, which has become a centerpoint of this literature.
\cite{TangY14} provides a simpler algorithm for computing $\EADAM$.
\cite{Reny22} introduces priority-neutral matchings, proves they form a lattice, and proves that the applicant-optimal matching in this lattice equals the outcome of $\EADAM$.

Other than priority-neutral matchings, likely the most significant prior theoretical framework for Pareto-improving the set of stable matchings is the set of \emph{legal} matchings \cite{EhlersM20}, which we denote by $\Legal$; see \autoref{footnote:legal} for a definition.
\cite{EhlersM20} proves that $\Legal$ forms a lattice, and the applicant-optimal matching in $\Legal$ equals the outcome of $\EADAM$.
\cite{Reny22} proves that every priority-neutral matching is legal (and that there can be legal, but not priority-neutral, matchings).
Thus, these two matching lattices are closely related. %
Throughout this paper, we exploit the connection between priority-neutrality and legality, along with additional structural results on legality proven by \cite{FaenzaZ22} (for more background, see \autoref{sec:aditional-prelims}).
However, we note that \cite{EhlersM20} show that $\Legal$ forms a \emph{distributive} lattice which is closed under coordinatewise minimum and maximum operations. 
In contrast, we prove the priority-neutral lattice is \emph{not} distributive, and hence is quite different in structure from the legal or stable lattice.

Many variations on stable matching have been proposed.
\cite{EhlersM20} follows the general idea of ``stable sets'' as in \cite{vonNeumannM44}; other papers taking similar or related approaches include \cite{ehlers2007neumann, mauleon2011neumann, atay2022school,dougan2023existence}, among others.
\cite{troyan2020essentially} define essentially stable matchings, which allow for Pareto-optimal matchings, but does not form a lattice (or characterize a unique Pareto-optimal matchings). 
\cite{abdulkadiroǧlu2020efficiency, kwon2020justified, tang2021weak} compare mechanisms by the sets of priorities they violate.
\cite{chen2023regret} studies a generalization of strategyproofness satisfied by $\EADAM$.
\cite{rong2023core} gives a novel characterization of $\EADAM$, as do many of the previously-mentioned papers.

In certain more-general matching problems, stable matchings are known to form a \emph{non-}distributive lattice. 
Specifically, this is the case in many-to-many \cite{Blair88} and many-to-one \cite{EcheniqueO04} matching settings with substitutes preferences; in fact, the recent paper \cite{EnF25} shows that \emph{all} finite lattices can arise as the stable lattice in this model.
These facts crucially use the generality of multi-partner matching with substitutes in preferences; 
indeed, when preferences are responsive (i.e., institutions preferences are consistent with a preference list over applicants) \cite{RothRuralHospital86}, or when they satisfy the commonly-assumed law of aggregate demand \cite{fleiner2003fixed, HatfieldM05}, stable matchings form a distributive lattice \cite{alkan2002class, li2014new}. 
We believe it is very interesting that priority-neutral lattices have a non-distributive and fundamentally new structure, despite being defined in the completely standard one-to-one matching model involving no new primitives.

Following \cite{Knuth74, GusfieldI89}, lattice properties of stable matchings have been instrumental in many results.
\cite{KarlinGW18,PalmerP21} make use the rotation poset to make breakthrough progress on the combinatorial problem of determining the maximum possible number of distinct stable matchings as a function of $n$, the number of applicants and institutions.
\cite{GonczarowskiT24} designs a blackboard verification (in the sense of NP) communication protocol for the applicant-optimal stable mechanism.
Much algorithmic work on different variants of such matching problems also exploits rotation poset representations and distributive lattice properties, see e.g. \cite{ManloveMatchPrefs13,echenique2023online}.
Additionally, \cite{IrvingL86, SabanS15} use related arguments to show that counting stable matchings, and other combinatorial matching problems, are $\#\mathsf{P}$-hard, and
\cite{Subramanian94, CookFL14} show that stable matchings are connected to a circuit model with ``comparator'' gates, and to certain novel complexity classes between $\mathsf{NL}$ and $\mathsf{P}$.
Further from our work,
\cite{PittelAverageStable89, ImmorlicaM05, AshlagiKL17, GimbertMM21, KanoriaMQ21, Mauras21, CaiT22} study stable matchings under certain random preference distributions.
Additionally, different (non-stable) matchings have been defined using ordinal preferences.
One example is popular matchings following \cite{abraham2007popular};
the structural properties of stable matchings have found applications for this set of matchings as well \cite{kavitha2024maximum}.
Another line of work investigates the extent to which natural decentralized matching dynamics converge to stable matchings \cite{roth1990random, ackermann2008uncoordinated, rudov2024fragile}.\footnote{
  To our knowledge, no such decentralized dynamics have been investigate that accommodate Pareto-optimal matchings or Pareto-improvements to stable matchings; this may be an interesting topic for future work.
}

\paragraph{Paper Organization.}

We recall the definition of stability and priority-neutrality in \autoref{sec:prelims}, and show that the priority-neutral lattice is not distributive in \autoref{sec:non-distributive}. 
In \autoref{sec:building}, we build up to our main characterization as follows: first, we recall additional preliminaries and results from prior work; second, we establish new structural properties of the priority-neutral lattice; and third, we provide a simpler sufficient representation of priority-neutral lattices (i.e., a representation which includes all priority-neutral lattices, but which includes other lattices as well).
In \autoref{sec:result}, give our main characterization result: a necessary and sufficient representation of priority-neutral lattices.
As a corollary, we also provide the first poly-time algorithm for checking whether a matching is priority-neutral.

\section{Preliminaries}
\label{sec:prelims}

\subsection{The Stable Matching Lattice}
\label{sec:prelims-stable}

We begin by recalling the definition of stable matchings and their essential properties.

\paragraph{Stabile Matchings and Priority Violations.}
We study one-to-one matchings between applicants $\Appls$ and institutions $\Insts$ (outside of some exposition in the introduction, where we referred to students and schools, respectively). 
We typically use notations like $\mu, \nu, \xi$ for matchings and (to avoid the common variables $a$ and $i$) we write $d\in \Appls, h \in \Insts$ (mnemonic: doctors and hospitals).
For $d \in \Appls$, we write $\mu(d) \in \Insts \cup \{ \emptyset \}$ for $d$'s match, where $\mu(d)=\emptyset$ denotes going unmatched (and likewise with $\mu(h)$ for $h \in \Insts$).

Agents on each side have a strict ordinal preference list over the other side, i.e., rankings $\succ_d$ over $\Insts$ for each $d\in\Appls$, and likewise $\succ_h$ over $\Appls$ for $h\in\Insts$. 
We refer to these rankings of applicants as their preferences, but by convention, we refer to the rankings of institutions as \emph{priorities} (reflecting the fact that, as in many school choice mechanisms, we think of the priorities as fixed and/or determined by policies, instead of reflecting the institution's preferences).
Following \cite{Reny22}, we assume for simplicity that each institution ranks all applicants;
we allow applicants' lists to be partial (i.e., they can rank some ``outside option'' $\emptyset$, which represents going unmatched, above some institutions).

\begin{definition}[Stable matchings and priority violations]
  A matching $\mu$ is \emph{unstable} if there exists a $d \in \Appls$ and $h\in\Insts$ such that $h \succ_d \mu(d)$ and $d \succ_h \mu(h)$. 
  In this case, we say that $d$'s \emph{priority is violated} (at $h$).
  If there are no such $d$ and $h$, we say that $\mu$ is \emph{stable}.
  Let $\Stab$ denote the set of all stable matchings.

  For a matching $\mu$, denote the set of applicants whose priority is violated in $\mu$ by 
  \[ \vio(\mu) = \{ d \in \Appls \ \mid\ \exists h\in \Insts:\text{ $ d\succ_h\mu(h)$ and $h \succ_d \mu(d)$} \}. \qedhere \]
\end{definition}

We let $\APDA$ denote the applicant-optimal stable matching, i.e. the outcome of the applicant-proposing deferred acceptance algorithm.
Likewise, let $\IPDA$ denote the institution-optimal / institution-proposing matching.
We typically consider a fixed profile of applicant preferences and institution priorities $P$, and thus suppress the dependence on $P$ in our notation $\APDA$, $\IPDA$, and other matching mechanism. 
When we need to refer to $P$ explicitly, we denote the matchings $\APDA(P)$ and $\IPDA(P)$.

\paragraph{(Stable) Lattices.}
This paper studies lattice structures on sets of matchings.
A (finite) \emph{lattice} $Q$ is a partial order $\le$ over a (finite) set $Q$ such that every pair of elements in $Q$ have a meet and a join.
The \emph{meet} (also called the \emph{greatest lower bound}) of $x, y \in Q$, denoted $x \wedge_Q y$, is the element $z \in Q$ such that $z\le x$ and $z \le y$, and moreover for each $w \in Q$ such that $w \le x$ and $w\le y$, we have $w \le z$.
Likewise, the \emph{join} (also called the \emph{least upper bound}) of $x, y \in Q$, denoted $x \vee_Q y$, is the element $z\in Q$ such that $z\ge x$ and $z \ge y$, and for each $w\in Q$ such that $w \ge x$ and $w\ge y$, we have $w \ge z$.
A lattice $Q$ is \emph{distributive} if, for each $x, y, z \in Q$, we have $\mu \wedge_Q (\nu \vee_Q \xi) = (\mu \wedge_Q \nu) \vee_Q (\mu \wedge_Q \xi)$.

We say that matching $\mu$ (weakly) \emph{dominates} matching $\nu$, denoted $\mu \ge \nu$, if for each applicant $d$ we have $\mu(d)\succeq_d \nu(d)$.
We write $\mu > \nu$ if $\mu \ge \nu$ and $\mu\ne\nu$.
Throughout this entire paper, the matchings are ordered by dominance.

For two matchings $\mu, \nu$, let the \emph{coordinatewise maximum} (resp., \emph{minimum}) of $\mu$ and $\nu$ be the function $\xi$ from applicants to institutions given by $\xi(d) = \max_{\succ_d} \{ \mu(d), \nu(d) \}$ (resp., $\xi(d)=\min_{\succ_d} \{ \mu(d), \nu(d) \}$). 
In some (but not all) of the latices we consider, joins/meets are given by coordinatewise maxima/minima.

It is classically known that stable matchings form a distributive lattice with joins/meets given by coordinatewise maxima/minima, and that all distributive lattices arise as the set of stable matchings \cite{Knuth76, blair1984every}.
For exposition, see \cite{GusfieldI89, echenique2023online}.

\begin{theorem}[Distributive Lattices Characterize Stable Matchings]
  The set of stable matchings is always nonempty, and forms a lattice under the dominance ordering.
  In this lattice, the join (resp., meet) of any two stable matchings is their coordinatewise maximum (resp., minimum).
  Moreover, a lattice $Q$ is isomorphic to a lattice of stable matchings for some set of preferences and priorities if and only if $Q$ is distributive.
\end{theorem}

\paragraph{Denoting Matchings Using Rotations.}

\label{sec:rotations-mentioned}

The distributive lattice of stable matchings additionally has a compact and useful representation in terms of a concept called ``rotations'' \cite{GusfieldI89,echenique2023online}.
We recall the full details of this representation in \autoref{sec:aditional-prelims}.
For ease of exposition, we begin by using rotations only as notation for denoting different matchings, as follows.

A \emph{rotation} is an ordered list of agents $\alpha = [ (h_1, d_1),\allowbreak (h_2, d_2),\allowbreak \ldots,\allowbreak (h_k, d_k)]$, where each $h_i \in \Insts$ and $d_i \in \Appls$.
If matching $\mu$ has $\mu(h_i)=d_i$ for each $i=1,\ldots,k$, then we let $\alpha\mu$, which is called the \emph{elimination of rotation $\alpha$ from $\mu$}, denote the matching which differs from $\mu$ only by matching $d_i$ to $h_{i+1}$ for $i=1,\ldots,k$ modulo $k$.
Hence, the action of a rotation can be visualized as moving the parentheses in $\alpha$ over one place, modulo $k$.
Moreover, we will always consider rotations $\alpha$ eliminated at matchings $\mu$ such that all applicants weakly prefer $\alpha\mu$ to $\mu$;
this convention allows us to think about matchings ``higher'' in the lattice as corresponding to ``larger'' sets of rotations, which we find intuitive.\footnote{
  In contrast to this, \cite{GusfieldI89} (resp., \cite{echenique2023online}) define rotations such that $\mu \ge \alpha \mu$, i.e., such that all men (resp., workers) weakly prefer $\mu$ to $\alpha\mu$.
  Other than this change in convention, the notions are identical.
  We adopt our convention---that larger sets of rotations correspond to matchings preferred more by the applicants---in an attempt to make working with priority-neutral matchings intuitive, since in this setting one side of the market (the applicants) are distinguished, and we seek to improve their welfare.
}
We write $\alpha_1\alpha_2\ldots\alpha_\ell\mu$ to mean $\alpha_1(\alpha_2(\ldots(\alpha_\ell\mu)\ldots))$; when the ``base'' matching $\mu$ is understood, we also simply denote this by $\alpha_1\alpha_2\ldots\alpha_\ell$.

\subsection{The Priority-Neutral Matching Lattice}
\label{sec:prelims-pn}

As discussed in the introduction, there are important reasons to look beyond the set of stable matchings. 
We first discuss a specific matching mechanism $\EADAM$ which improves applicants' matches beyond their stable ones, and will be important for our analysis.
Then we recall the main notion we are interested in: priority-neutrality.

\paragraph{Pareto-Optimality and $\EADAM$.}

When $\mu$ dominates $\nu$, i.e., $\mu > \nu$, we also say that $\mu$ Pareto-improves $\nu$.
Matching $\mu$ is \emph{Pareto-optimal} if there does not exist a $\nu$ with $\nu > \mu$.
The literature on Pareto-improving stable matchings has identified a central mechanism which is Pareto-optimal and dominates $\DA$, namely, \cite{kesten2010school}'s efficiency-adjusted DA mechanism ($\EADAM$).
One approach to defining $\EADAM$, pioneered by \cite{TangY14}, is as follows:

\begin{definition}[$\EADAM$ and the Kesten-Tang-Yu sequence]
  \label{def:Kesten-Tang-Yu-Alg}
  For a profile of preferences and priorities $P$, define the Kesten-Tang-Yu sequence as follows. 
  Let $\mu_1$ denote the outcome of $\APDA(P)$ under preference and priorities $P = P_1$. Call an institution $h$ \emph{under-demanded in round $i$} if $h$ never rejects any applicant during the run of $\APDA(P_i)$ (additionally, the ``institution'' $\emptyset$ representing going unmatched is always considered under-demanded). 
  For each $i \ge 2$, let $P_i$ denote modifying the profile of preferences as follows: for each applicant $d$ matched in $\mu_{i-1}$ to an institution $h$ which is under-demanded in round $i-1$, remove all institutions strictly before $h$ in $d$'s preference list.
  Then, let $\mu_i = \APDA(P_i)$.
  There must exist an $N$ such that $\mu_n = \mu_{N}$ for each $n \ge N$.
  The outcome of efficiency-adjusted DA ($\EADAM$) is $\mu_N$.
\end{definition}

\cite{TangY14} prove that at least one institution is under-demanded in each round, and thus $\EADAM$ can be computed via at most $N$ runs of $\APDA$, where $N$ is the number of institutions.
Moreover, $\EADAM$ is Pareto-optimal and dominates all stable matchings.
$\EADAM$ is closely related to our main definition of interest, priority-neutrality, and will feature into our analysis.

\paragraph{Priority-Neutrality.}

We now turn to the main object of study in this paper: the lattice of priority-neutral matchings \cite{Reny22}.
In words, a matching $\mu$ is priority-neutral if it's impossible to make any applicant whose priority is violated in $\mu$ better off, without violating the priority of some applicant who is made worse off.
To fully spell out this definition, we introduce the terminology ``priority-correcting adjustment'' to aid in discussing matchings which are \emph{not} priority-neutral.
Namely, we have:
\begin{definition}[Priority-Neutrality]
  \label{def:priority-neutral}
  Say that matching $\nu$ is a \emph{priority-correcting adjustment (abbreviated PCA)} of matching $\xi$ if:
  \begin{enumerate}[(1)]
    \item \label{item:vio-impr} There exists an applicant $d_* \in \vio(\xi)$ with $\nu(d_*) \succ_{d_*} \xi(d_*)$.
    \item \label{item:decr-nonvio} For each $d$ such that $d \in \vio(\nu)$, we have $\nu(d) \succeq_d \xi(d)$.
  \end{enumerate}

  A matching $\mu$ is \emph{priority-neutral} if there does not exist any priority-correcting adjustment of $\mu$. 
  We denote the set of priority-neutral matchings by $\PN$.
\end{definition}

Priority-correcting adjustments capture failures of priority-neutrality; indeed, \cite{Reny22} argues they ``correct'' some applicant $d$'s priority violation (without unduly harming others).
In other words, priority-neutral matchings allow for some applicants' priority to be violated, but only when such applicants cannot be improved without harming others and also violating their priority.

Observe that, since stable matchings violate no applicant's priority, every stable matching is priority-neutral.
Thus, $\PN$ generalizes the set of stable matchings, and \cite{Reny22} proves that $\PN$ also forms a lattice.
Moreover, \cite{Reny22} proves that the set $\PN$ is closed under the coordinatewise maximum operation, and identifies $\IPDA$ and $\EADAM$ as the applicant-pessimal and applicant-optimal elements of $\PN$, respectively.
Formally:

\begin{theorem}[\cite{Reny22}]
  \label{thrm:reny-lattice}
  For any profile of priorities and preferences, $\PN$ is a lattice under $\le$, and this lattice contains every stable matching.
  Additionally, the applicant-pessimal element of $\PN$ is $\IPDA$, and the applicant-optimal element of $\PN$ is $\EADAM$.
  Finally, for any two matchings $\mu, \nu$ in $\PN$, the least upper bound of $\mu$ and $\nu$ in $\PN$ equals their coordinatewise maximum.
\end{theorem}

Note, however, that \cite{Reny22} leaves as an open question whether the greatest lower bound of two matchings in $\PN$ is their coordinatewise minimum.
Our first result below resolves this question.

\section{Priority-Neutral Lattices Are Not Distributive}
\label{sec:non-distributive}

In this section, we show that the $\PN$ lattice lacks much of the tractable structure of the stable lattice.
We give a self-contained proof that $\PN$ is not distributive, and that greatest lower bounds in $\PN$ are not given by coordinatewise minimums.

We prove that the above results via an explicit example, presented in \autoref{fig:non-distributive}.
\autoref{fig:non-distributive-a} gives the priorities and preferences of the construction.
As we show below, \autoref{fig:non-distributive-b} displays the $\PN$ lattice for this instance, where $\mu\le\nu$ if any only if $\mu$ is below $\nu$ in a path in the diagram. %
\autoref{fig:non-distributive-c} displays $\IPDA$ and the set of rotations used to describe $\PN$, as discussed in \autoref{sec:rotations-mentioned}.
For brevity, through this section we write $\alpha_1\ldots\alpha_k$ to mean $\alpha_1\ldots\alpha_k\IPDA$.%
\footnote{
    For example, the matching denoted $\alpha$ equals $\{(h_A, d_B), (h_B, d_A), (h_C, d_C), (h_D, d_D), (h_E, d_E), (h_F, d_F), (h_Z, d_Z)\} $, and the matching $\beta \gamma$ equals $\{(h_A, d_A), (h_B, d_B), (h_C, d_D), (h_D, d_C), (h_E, d_F), (h_F, d_E), (h_Z, d_Z)\}$.
}

    \begin{figure}[htb]
        \begin{minipage}{0.45\textwidth}
            \begin{minipage}{\textwidth} \centering
                Institutions' Priorities:
                \\[0.05in]
                \begin{tabular}{ccccccc}
                    \toprule
                    $h_A$ & $h_B$
                    & $h_C$ & $h_D$
                    & $h_E$ & $h_F$
                    & $h_Z$ \\
                    \midrule
                    $d_A$ & $d_B$
                    & $d_C$ & $d_D$
                    & $d_E$ & $d_F$
                    & $d_Z$ \\
                    $d_B$ & $d_A$
                    & $d_Z$ & $d_C$
                    & $d_F$ & $d_E$
                    & $d_A$ \\
                    $d_E$ & $\myvdots$ 
                    & $d_D$ & $\myvdots$
                    & $d_Z$ & $\myvdots$ 
                    & $\myvdots$ \\
                    $\myvdots$ &
                    & $\myvdots$ &
                    & $\myvdots$ &
                    &
                \end{tabular}
            \end{minipage}
            \\[0.3in]
            \begin{minipage}{\textwidth} \centering
                Applicants' Preferences:
                \\[0.05in]
                \begin{tabular}{ccccccc}
                    \toprule
                    $d_A$ & $d_B$
                    & $d_C$ & $d_D$
                    & $d_E$ & $d_F$
                    & $d_Z$ \\
                    \midrule
                    $h_B$ & $h_A$
                    & $h_D$ & $h_C$
                    & $h_F$ & $h_E$
                    & $h_E$ \\
                    $h_Z$ & $h_B$
                    & $h_C$ & $h_D$
                    & $h_A$ & $h_F$
                    & $h_C$ \\
                    $h_A$ & $\myvdots$ 
                    & $\myvdots$ & $\myvdots$
                    & $h_E$ & $\myvdots$ 
                    & $h_Z$ \\
                    $\myvdots$ & 
                    & &
                    & $\myvdots$ &
                    & $\myvdots$
                \end{tabular}
            \end{minipage} 
            \begin{minipage}{\textwidth} 
                {\centering
                \subcaption{Preferences and Priorities.}
                \label{fig:non-distributive-a}
                }
                {\footnotesize
                \textbf{Note:} The vertical dots indicate that the remainder of the preference/priority list is irrelevant.
                \par}
            \end{minipage} 
        \end{minipage} 
        \qquad
        \begin{minipage}{0.45\textwidth}
                \includegraphics[width=0.8\textwidth]{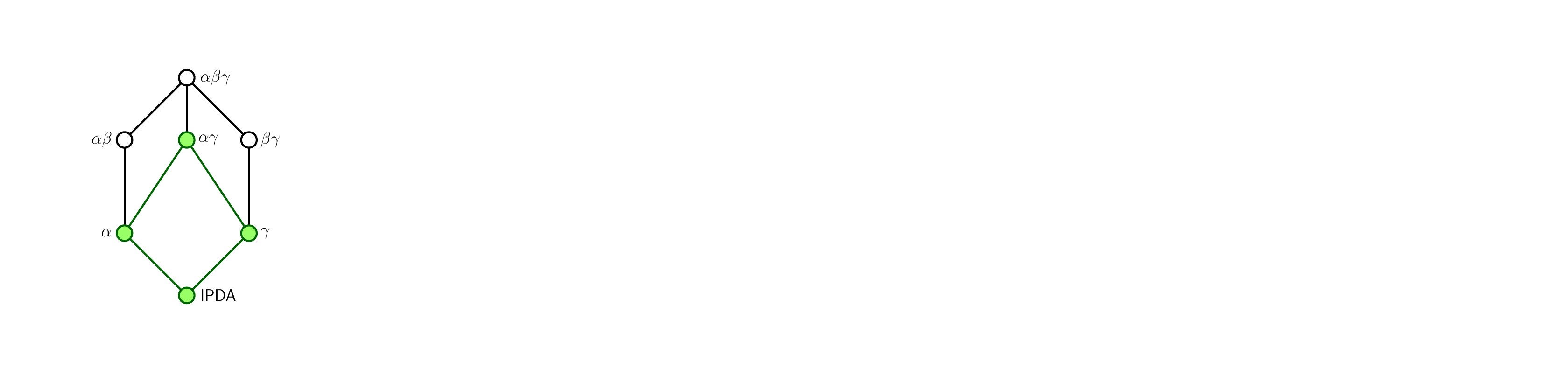}
            \begin{minipage}{\textwidth} \centering
                \subcaption{Priority-neutral lattice.}
                \label{fig:non-distributive-b}
            \end{minipage} 
        \end{minipage} 
        \\[0.5in]
        \begin{minipage}{0.85\textwidth} \centering
            \begin{minipage}{0.5\textwidth} \centering
                Institution-optimal stable matching $\mu_0$:
                \\[0.05in]
                \begin{tabular}{ccccccc}
                    $h_A$ & $h_B$
                    & $h_C$ & $h_D$
                    & $h_E$ & $h_F$
                    & $h_Z$ \\
                    \midrule
                    $d_A$ & $d_B$
                    & $d_C$ & $d_D$
                    & $d_E$ & $d_F$
                    & $d_Z$
                \end{tabular}
            \end{minipage} 
            \quad
            \begin{minipage}{0.45\textwidth} \centering
                Relevant rotations:
                \begin{align*}
                    \alpha & = [ (h_A, d_A), (h_B, d_B) ]
                    \\ \beta & = [ (h_C, d_C), (h_D, d_D) ]
                    \\ \gamma & = [ (h_E, d_E), (h_F, d_F) ]
                \end{align*}
            \end{minipage} 
            \subcaption{Matching and rotations used to denote the matchings in \autoref{fig:non-distributive-b}.}
            \label{fig:non-distributive-c}
        \end{minipage} 
        \caption{Set of priorities and preferences witnessing \autoref{thrm:non-distributive}, along with the priority-neutral lattice (with the stable matchings shaded in green).}
        \label{fig:non-distributive}
    \end{figure}

The crucial fact about this construction is that $\alpha\beta$ and $\beta\gamma$ are priority-neutral, but $\beta$ is not priority-neutral. 
This will immediately imply that $\PN$ is not closed under coordinatewise minimums (since indeed $\beta$ is the coordinatewise minimum of $\alpha\beta$ and $\beta\gamma$).
Two important things are used to guarantee this property.
First, there is an applicant $d_Z$ can be Pareto-improved in $\IPDA$;
however, either of the rotations $\alpha$ or $\gamma$ make this Pareto-improvement no longer possible.\footnote{
    The most involved part of the proof below involves arguing that in $\alpha\beta$ and $\beta\gamma$, not only is $d_Z$ not Pareto-improvable, but $d_Z$ has no priority-correcting adjustment (\autoref{def:priority-neutral}).
}
Second, $d_Z$'s priority is violated in $\beta$, %
and additionally $d_Z$ can Pareto-improve in $\beta$ (as well as $\IPDA$), showing that $\beta \notin \PN$.
As we argue below, all of the other relevant combinations of rotations will be priority-neutral matchings, showing that $\PN$ is given by the lattice in \autoref{fig:non-distributive}.

\begin{proposition}
    \label{prop:non-distributive-proof}
    For the preferences and priorities given in \autoref{fig:non-distributive-a}, priority-neutral lattice $\PN$ is given by \autoref{fig:non-distributive-b}.
\end{proposition}
\begin{proof}
    To begin, we check that the stable matching lattice, which is always contained in $\PN$ by \autoref{thrm:reny-lattice}, is given by the green-shaded matchings in \autoref{fig:non-distributive-b}.
    First, the institution-optimal matching equals $\mu_0$, as defined in \autoref{fig:non-distributive-c}, since all institutions have distinct top-priority applicants equal to their match in $\mu_0$.
    Second, one can run applicant-proposing deferred acceptance to see that $\alpha\gamma$ is the applicant-optimal stable matching. %
    Third, one can verify directly that the matchings $\alpha$ and $\gamma$ have no priority violations, and are thus stable; since no other distinct matchings $\mu'$ exist with $\mu_0 < \mu' < \alpha\gamma$, these green shaded matchings are all stable matchings.\footnote{
        In full detail, in $\alpha$, all institutions except for $h_A$ and $h_B$ are matched to their top-priority applicants; while $d_E$ prefers $h_A$ to their match in $\alpha$, we have $d_B \succ_{h_A} d_E$, so $\alpha$ is stable.
        Similarly, in $\gamma$, all institutions except for $h_E$ and $h_F$ are matched to their top-priority applicants; while $d_Z$ prefers $h_E$ to their match in $\gamma$, we have $d_F \succ_{h_E} d_Z$, and $\gamma$ is stable.
        Alternatively, the methods of \cite{GusfieldI89} can be used to directly compute the set of rotation representing the stable matching lattice (which are $\alpha, \gamma$ with no predecessor relations).
    }

    Next, we verify that the applicant-optimal priority-neutral matching is $\alpha\beta\gamma$.
    By \autoref{thrm:reny-lattice}, this matching equals the $\EADAM$ outcome, which can be shown to equal $\alpha\beta\gamma$ by the procedure in \autoref{def:Kesten-Tang-Yu-Alg}.
    (Indeed, in the second round of the Kesten-Tang-Yu sequence, $d_Z$ will be fixed at $h_Z$, and $d_C$ and $d_D$ will match to their top choices $h_D$ and $h_C$, so $\alpha\beta\gamma$ equals $\EADAM$ and this matching is the highest element of the $\PN$ lattice.)

    Now, three relevant matchings remain which we have not yet considered: $\beta$, $\alpha\beta$, and $\beta\gamma$.
    Since these are the only not-yet-considered matchings $\mu$ such that $\mu_0 < \mu < \alpha\beta\gamma$, we know by \autoref{thrm:reny-lattice} that all non-stable matchings in $\PN$ must be in $\{\beta, \alpha\beta, \beta\gamma \}$.
    We will proceed by showing that $\beta \notin \PN$, but $\alpha\beta \in \PN$ and $\beta\gamma \in \PN$.

    We now show $\beta \notin \PN$. To see this, consider $\beta$, and define the matching $\nu$ as follows:
    \\[0.1in]
    \begin{minipage}{\textwidth} \centering
    $\beta:$
    \begin{tabular}{ccccccc}
        $h_A$ & $h_B$
        & $h_C$ & $h_D$
        & $h_E$ & $h_F$
        & $h_Z$ \\
        \midrule
        $d_A$ & $d_B$
        & $d_D$ & $d_C$
        & $d_E$ & $d_F$
        & $d_Z$
    \end{tabular}
    \qquad
    \qquad
    $\nu_*: $
    \begin{tabular}{ccccccc}
        $h_A$ & $h_B$
        & $h_C$ & $h_D$
        & $h_E$ & $h_F$
        & $h_Z$ \\
        \midrule
        $d_E$ & $d_B$
        & $d_D$ & $d_C$
        & $d_Z$ & $d_F$
        & $d_A$
    \end{tabular}.
    \label{page:nu-star-beta-PCA}
    \end{minipage}
    \\[0.1in]
    We claim that $\nu_*$ is a priority-correcting adjustment of $\beta$ for $d_Z$. 
    To see this, note that in $\beta$, applicant $d_Z$'s priority is violated (at institutions $h_C$), and $\nu_*(d_Z) \succ_{d_Z} \beta(d_Z)$.
    Additionally, note that $\nu_*$ Pareto-dominates $\beta$, i.e., no applicant $d$ exists with $\beta(d) \succ_d \nu_*(d)$ (regardless of whether the priority of $d$ is violated).
    Thus, $\nu_*$ is a priority-correcting adjustment of $\beta$, and $\beta\notin\PN$.

    Next, consider the matchings $\alpha\beta$ and $\beta\gamma$:
    \\[0.1in]
    \begin{minipage}{\textwidth} \centering
    $\alpha\beta:$
    \begin{tabular}{ccccccc}
        $h_A$ & $h_B$
        & $h_C$ & $h_D$
        & $h_E$ & $h_F$
        & $h_Z$ \\
        \midrule
        $d_B$ & $d_A$
        & $d_D$ & $d_C$
        & $d_E$ & $d_F$
        & $d_Z$
    \end{tabular}
    \qquad
    \qquad
    $\beta\gamma: $
    \begin{tabular}{ccccccc}
        $h_A$ & $h_B$
        & $h_C$ & $h_D$
        & $h_E$ & $h_F$
        & $h_Z$ \\
        \midrule
        $d_A$ & $d_B$
        & $d_D$ & $d_C$
        & $d_F$ & $d_E$
        & $d_Z$
    \end{tabular}.
    \end{minipage}
    \\[0.1in]
    We will show that both $\alpha\beta$ and $\beta\gamma$ are in $\PN$, i.e., we will show that these matchings have no priority-correcting adjustments.
    To begin to see this, observe directly that $d_Z$ is the only applicant whose priority is violated in these matchings, so any priority-correcting adjustment $\nu$ must improve the match of $d_Z$.
    Informally, we will show that no priority-correcting adjustment can exist because the only ``reasonable'' candidate will be $\nu_*$ as defined above. 
    However, $\nu_*$ makes $d_B$ worse off than in $\alpha\beta$, and $d_B$'s priority is violated at $h_A$ in $\nu_*$, so $\nu_*$ is not a priority-correcting adjustment of $\alpha\beta$.
    Similarly, $\nu_*$ makes $d_F$ worse off than in $\beta\gamma$, and $d_F$'s priority is violated at $h_E$ in $\nu_*$, so $\nu_*$ is not a priority-correcting adjustment of $\beta\gamma$.
    Below, we show that indeed $\alpha\beta$ and $\beta\gamma$ have \emph{no} priority-correcting adjustments.

    We now formally show that $\alpha\beta \in \PN$.
    Suppose for contradiction that there exists a priority-correcting adjustment $\nu$ of $\alpha\beta$.
    Since the only applicant in $\alpha\beta$ whose priority is violated in $d_Z$, we know that $\nu$ must improve the match of $d_Z$, so $\nu(d_Z) = h_C$ or $\nu(d_Z) = h_E$.
    We consider these two cases.
    In the first case, if $\nu(d_Z) = h_C$, then $\nu(d_D)$ must be $h_D$ (since $d_D$ must receive a worse match in $\nu$ than in $\alpha\beta$, where $\alpha\beta(d_D) = h_C$, and $d_D$ has highest-priority at $h_D$, which is $d_D$'s top choice after $h_C$).
    But, then $\nu(d_C)$ must be $h_C$ (since similarly $d_C$ must receive a worse match in $\nu$ than in $\alpha\beta$, and $d_C$ has top priority at $h_C$), however, this contradicts the assumption that $d_Z$ is matched to $h_C$ in matching $\nu$.
    In the second case, suppose $\nu(d_Z) = h_E$.
    Since $d_E$ has higher priority than $d_Z$ at $h_E$, we know $d_E$ must be matched in $\nu$ to an institution she prefers to $h_E$; so we either have $\nu(d_E)=h_A$ or $\nu(d_E)=h_F$.
    If $\nu(d_E)=h_A$, then $d_B$ must receive a worse match in $\nu$ than in $\alpha\beta$; since $d_B$ has higher priority at $h_A$ than $d_E$, this would mean $\nu$ cannot be a priority-correcting adjustment, so we cannot have $\nu(h_E)=h_A$.
    If $\nu(d_E)=h_F$, then $\nu(d_F)$ must be different from $\alpha\beta(d_F) = h_F$; since $d_F$ has higher priority at $h_E$ than $d_Z$ and $d_F$'s preference list is $h_E\succ h_F$, this makes it impossible that $\nu(d_Z) = h_E$, and we cannot have $\nu(d_E)=h_F$.
    Thus, this shows that there exists no priority-correcting adjustment $\nu$ of $\alpha\beta$, so $\alpha\beta\in\PN$.

    We now formally show that $\beta\gamma \in \PN$. 
    Again, assume for contradiction there is a priority-correcting adjustment $\nu$ of $\beta\gamma$.
    Since $d_Z$ is the only applicant whose priority is violated in $\beta\gamma$, we must have $\nu(d_Z) = h_C$ or $\nu(d_Z) = h_E$.
    By precisely the same logic as in the previous paragraph, we cannot have $\nu(d_Z) = h_C$ (since this case only considered the matches of $d_C$ and $d_D$, which are the same in $\beta\gamma$ as in $\alpha\beta$).
    Thus, consider $\nu$ such that $\nu(d_Z) = h_E$.
    However, since $d_F$'s favorite institution was $h_E = \beta\gamma(d_F)$, and $d_F$ has higher priority than $d_Z$ at $h_E$, we immediately get that $\nu$ is not a priority-correcting adjustment of $\beta\gamma$.
    Thus, this shows that there exists no priority-correcting adjustment $\nu$ of $\beta\gamma$, so $\beta\gamma\in\PN$.

    Thus, the priority-neutral lattice consists of exactly those matchings in \autoref{fig:non-distributive-b}, completing the proof.
\end{proof}

We can now prove our first theorem.
The above example shows that the greatest lower bound of two matchings in $\PN$ need not be their coordinatewise minimum,
resolving an open question of \cite{Reny22}.
It also shows that the $\PN$ lattice is not distributive.

\begin{theorem}[$\PN$ Is Not Distributive]
    \label{thrm:non-distributive}
    $\PN$ is not a distributive lattice.
    Moreover, the greatest lower bound of two matchings in $\PN$ is not their coordinatewise minimum.
\end{theorem}
\begin{proof}
    We prove that both results of the theorem hold true in the lattice of priority-neutral matchings in \autoref{fig:non-distributive}.
    To prove the second sentence, observe that
    \[ \alpha\beta\wedge_{\PN}\beta\gamma = \mu_0 \ne \beta = \alpha\beta \wedge_{\Legal} \beta\gamma,\] 
    where we write $\wedge_{\PN}$ to denote the greatest lower bound of two matchings within $\PN$, and $\wedge_{\Legal}$ to denote their coordinatewise minimum.
    To prove the first sentence, we build on the above fact (used in equality (*)) and observe that
    \[ \alpha\beta \wedge_{\PN} \big(\alpha \vee \beta\gamma\big)
        = \alpha\beta \wedge_{\PN} \alpha\beta\gamma
        = \alpha\beta
        \ne \alpha
        = \alpha \vee \mu_0
        \overset{(*)}{=} \big(\alpha\beta \wedge_{\PN} \alpha \big) \vee \big( \alpha\beta \wedge_{\PN} \beta\gamma\big).
    \]
    Thus, the $\PN$ lattice is not distributive, which finishes the proof.
\end{proof}

\section{Building Up To Our Main Result}
\label{sec:building}

In this section, we build up to our main result.
First, in \autoref{sec:aditional-prelims}, we recall more-detailed preliminaries and prior results which will be important for our analysis.
Second, in \autoref{sec:structural-properties}, we prove new structural results of $\PN$, particularly a novel ``good chain'' property.
Third, in \autoref{sec:simpler-sufficient}, we prove a partial characterization, namely a necessary condition that all $\PN$ lattices must satisfy, as a tool to build towards our main characterization and illustrate the intricacies involved.

\subsection{Additional Preliminaries: Rotation DAG and Legal Lattice}
\label{sec:aditional-prelims}

\subsubsection{Rotation DAG Representation of \texorpdfstring{$\Stab$}{Stab}}
\label{sec:aditional-prelims-rotation-DAG}
Following the classical theory of stable matchings, we represent matching lattices via subsets of rotations (introduced as notation in 
\autoref{sec:rotations-mentioned}).
We recall here all the properties and definitions we need,
but defer full coverage of some topics (such as algorithms for actually finding the rotation representation) to the standard reference \cite{GusfieldI89}.
See also \cite{CaiT19} or \cite[Section 1.4]{echenique2023online}.\footnote{
  While the notion in our definitions are identical to classical representations, we phrase the definitions with slightly more data (such as having a fixed ordering on the rotations that gives a topological sort of the DAG).
  We find these extra pieces of data conceptually helpful when we generalize this construction in later sections, since they help express additional conditions on these generalized constructions.
}

Recall that %
a rotation $\rho = [ (h_1, d_1),\allowbreak \ldots,\allowbreak (h_k, d_k)]$, where each $h_i\in\Insts$ and $d_i\in\Appls$, is used to represent the difference between two matchings.
If $\mu$ has $\mu(h_i)=d_i$ for each $i$, we say $\rho$ is \emph{valid} at $\mu$.
In this case, we let $\alpha\mu$, called the \emph{elimination of $\alpha$ from $\mu$}, denote the matching which differs from $\mu$ only by matching $d_i$ to $h_{i+1}$ for $i=1,\ldots,k$ modulo $k$.
(We also denote rotations with letters such as $\alpha,\beta,\gamma$ in specific examples.)

\begin{definition}[Rotation DAG of distributive lattices]
  \label{def:rotation-dag}
    For some set of applicants $\Appls$ and institutions $\Insts$, a \emph{rotation DAG} is a tuple $G = (\mu, R, P)$, defined as follows.
    First, $\mu$ is some matching, referred to as the ``base matching.''
    Second, $R$ is some set of rotations $R = \{ \rho_1, \ldots, \rho_N \}$,
    equipped with a fixed ordering $[\rho_1,\ldots,\rho_N]$.
    Third, $P$ is a set of \emph{predecessor relationships} over $R$. 
    Each such predecessor relationship is a pair $p = (\rho_i, \rho_j) \in R \times R$, where $i > j$; we also think of such a $p$ as an edge $\rho_i \to \rho_j$, and in this case we say that $\rho_j$ is a predecessor of $\rho_i$ (and the fact that we always have $i>j$ means that $(R,P)$ forms a DAG).

    Consider some subset $X\subseteq R$ and some predecessor relationship $p = (r_0, r_1) \in P$.
    We say that $p$ is \emph{violated} by $X$ if $r_0 \in X$, but $r_1 \notin X$.\footnote{
      This terminology is inspired by the fact that $S\subseteq R$ represents a matching $\mu$, and whenever some predecessor relationship $p$ is violated by $S$, this means that some applicant's priority is violated in $\mu$.
    }

    Define the \emph{(abstract) distributive lattice}, denoted $D_{\mathsf{abs}}(G)$, to be the collection of subsets of $R$, ordered by set inclusion, defined as follows:
    \[ D_{\mathsf{abs}}(G) = \left\{ X \subseteq R\ \big|\ \text{No $p\in P$ is violated by $X$} \right\}. \]

    Define the \emph{(matching) distributive lattice}, denoted $D(G)$, as follows.
    For every $X \in D_{\mathsf{abs}}(G)$, we require that if $X = \{ \rho_{i(1)}, \ldots \rho_{i(k)} \}$, where $i(1) < \ldots < i(k)$, then $\rho_{i(j+1)}$ is valid at $\rho_{i(j)}\ldots\rho_{i(1)}\mu$ for each $j=0,\ldots,k-1$.
    Then, define $\match(X) = \rho_{i(k)}\ldots\rho_{i(1)}\mu$.
    We set
    \[ D(G) = \left\{ \match(X)\ \big|\ \text{$X \subseteq R$, and no $p\in P$ is violated by $X$} \right\}. \qedhere \]
\end{definition}

Note that any set $R$ and predecessor DAG on $R$ can define an abstract lattice $D_{\mathsf{abs}}(G)$ on subsets of $R$.
The matching lattice $D(G)$ then specifies, for each $X \in D_{\mathsf{abs}}(G)$, a corresponding matching $\match(X)$ which arises by eliminating all of the rotations in $X$ from the ``base'' matching $\mu$.
For a matching $\mu \in D(G)$, we also define $\rots(\mu) = X \in D_{\mathsf{abs}}(G)$ to be the unique $X \subseteq R$ such that $\match(X) = \mu$.

The classical representation of $\Stab$ is as follows:

\begin{theorem}[\cite{GusfieldI89}]
  \label{thrm:stable-rotation-dag}
  There exists a set of rotations $R$ and predecessors $P$, such that if $G = (\IPDA, R, P)$, then $D(G) = \Stab$.
\end{theorem}

Additionally, \cite{GusfieldI89,CaiT19} show that we can identify the rotations and predecessors as follows. 
Recall that in any lattice, we say $\mu$ covers $\nu$, denoted $\mu \covers \nu$, if $\mu > \nu$ and for all $\xi$ with $\mu \ge \xi \ge \nu$, we have $\xi = \mu$ or $\xi = \nu$.
The set of rotations $R$ is the set of all $\rho$ such that there exists a $\mu \in \Stab$ where $\rho$ is valid in $\mu$, and $\rho\mu \in \Stab$, and
where $\rho\mu$ covers $\mu$ in $\Stab$; intuitively, this represents the ``minimal possible differences'' between two stable matchings.
(Note in particular that rotations always improve that match of applicants, and worsen the match of institutions.)
In this case, we say that $\rho$ is \emph{exposed} in $\mu$.\footnote{
  Note that even if $\rho$ is valid in $\mu$, it may not be the case that $\rho$ is exposed in $\mu$.
}
The predecessors $P$, which come in two different types, can be identified as follows:
\begin{enumerate}
  \item Rotation $\rho_1$ is a \emph{type 1 predecessor} of rotation $\rho_2$ if there exists a $(d,h)\in\Appls\times\Insts$ such that $\rho_1$ moves $d$ to $h$, and $\rho_2$ moves $d$ above $h$.  In other words, we have:
  \begin{align*}
       & \rho_2 = [ \ldots, (h, d), \ldots]
       && \qquad \text{for all such rotations.}
    \\ & \downarrow \text{  (type 1) }
    \\ & \rho_1 = [ \ldots, (\cdot, d), (h, \cdot), \ldots]
  \end{align*}
  \item Rotation $\rho_1$ is a \emph{type 2 predecessor} of rotation $\rho_2$ if there exists a $(d,h)\in\Appls\times\Insts$ such that $\rho_1$ moves $d$ from below $h$ to above $h$, and $\rho_2$ moves $h$ from above $d$ to below $d$.
  In other words, we have:
  \begin{align*}
       & \rho_2 = [ \ldots, (\cdot, d_{\text{down}}), (h, d_{\text{up}}), \ldots]
       && \qquad \text{if we have $h_{\text{up}} \succ_d h \succ_d h_{\text{down}}$, }
    \\ & \downarrow \text{  (type 2) }
       && \qquad \text{and also $d_{\text{up}} \succ_h d \succ_h d_{\text{down}}$.}
    \\ & \rho_1 = [ \ldots, (h_{\text{down}}, d), (h_{\text{up}}, \cdot), \ldots]
  \end{align*}
\end{enumerate}

Intuitively, type 1 predecessor relations are needed because applicants can only ``switch stable partners'' in the units specified by rotations, and because if such a $d$ is improved above $h$ by $\rho_2$, then $d$ must have already be improved \emph{to} $h$ by $\rho_1$.
It is easy to see why type 2 predecessor relations are needed: if we eliminate such a $\rho_2$ without eliminating $\rho_1$, then $d$'s priority at $h$ is violated, so the matching is not stable.
For example, the stable matching lattice in \autoref{fig:simple-ex} is represented by rotations $\alpha = [ (h_A,d_A), (h_B,d_B)]$ and $\beta = [ (h_C,d_C), (h_D,d_D)]$, where $\alpha$ is a type 2 predecessor of $\beta$.
See \cite{GusfieldI89} or \cite{CaiT19} for details.
We build on the above concrete representation of the rotations and predecessors in \autoref{sec:result} when we prove our full characterization.

\subsubsection{The Legal Matching Lattice}
\label{sec:extra-prelims-legal}

While our main focus is on the priority-neutral lattice $\PN$, we make use of another interesting generalization of $\Stab$ known as the \emph{legal} lattice \cite{EhlersM20} (already introduced in \autoref{footnote:legal}), which will aid us greatly in our analysis. 
A set $\Legal$ of matchings is called \emph{legal} if
\[\Legal = \{ \mu\ \allowbreak|\ \allowbreak \text{$\forall \nu \in \Legal:$ \allowbreak no applicant $d$'s priority is violated by $\mu$ at institution $\nu(d)$} \}. \]
\cite{EhlersM20} prove that there exists a unique legal set $\Legal$, which always contains $\Stab$,
and that the applicant-optimal element of $\Legal$ is $\EADAM$.
Thus, $\PN$ and $\Legal$ have the same topmost element; moreover, \cite{Reny22} proves that $\PN\subseteq \Legal$, so these two lattices are closely related.

In contrast to $\PN$, however, $\Legal$ has the same tractable structure as $\Stab$.
\cite{FaenzaZ22} formalizes this by showing that every legal lattice can be represented as the stable lattice on some \emph{altered} set of preferences and priorities. 
This allows us to greatly simply our analysis (intuitively, by representing $\PN$ as a subset of $\Legal$, where we use the ``off the shelf'' rotation DAG of $\Legal$, and specify which legal matchings are not priority-neutral).
In summary, the main known results are the following:

\begin{theorem}[\cite{EhlersM20, Reny22, FaenzaZ22}]
  \label{thrm:pn-in-l-l-lattice}
    We have $\PN\subseteq \Legal$, and $\Legal$ forms a distributive lattice which can be represented via a partial order on a set of rotations over some other set of preferences and priorities.
\end{theorem}

For convenience, in most of our analysis we actually use the sub-lattice $\K$ of $\Legal$ which also has the same bottommost element as $\PN$, namely $\IPDA$.\footnote{
  To see that $\Legal$ does not in general have the same lowest elements as $\PN$, 
  observe that the definition of $\Legal$ is symmetric in the roles of applicants and institutions, and thus $\Legal$ always contains the institution-side analogue of $\EADAM$ (which is a different matching from $\IPDA$).
}
In other words, we define $\K = \{ \mu \in \Legal \ \mid\ \mu \ge \IPDA \}$; note that $\PN\subseteq \K$. 
It is easy to see that $\K$ is also distributive. 
See \autoref{fig:all-lattices} for an illustration of all matching lattices we consider.

\begin{figure}[htbp]
    \begin{minipage}[t]{0.6\textwidth}
      \strut\vspace*{-\baselineskip}\newline
      \includegraphics[width=\textwidth]{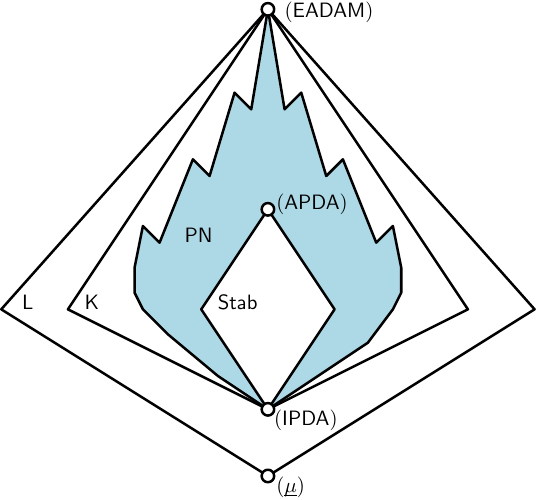}
    \end{minipage}
    \qquad
    \begin{minipage}[t]{0.35\textwidth}
      \phantom{.}

      \caption{Illustration of matching lattices.}
      \label{fig:all-lattices}

      \phantom{.}

      \begin{minipage}{\textwidth}  
        \footnotesize
        \textbf{Notes:}
        As the figure shows, $\Stab\subseteq \PN\subseteq \K\subseteq \Legal$.
        Each of $\Stab, \K$, and $\Legal$ are always distributive lattices, and can be represented with a predecessor DAG on a set of rotations, and hence we illustrate each of these with a straight-bordered shape.
        In contrast, as shown in \autoref{sec:non-distributive}, $\PN$ is not distributive.
        The applicant optimal / pessimal matchings in each lattice are labeled within parentheses, where the matching $\underline{\mu}$ denotes the result of running $\EADAM$ on the instance produced by interchanging the roles of the applicants and the institutions.
      \end{minipage}
    \end{minipage}
\end{figure}

The fact that the lattice $\K$ has a rotation DAG representation is guaranteed by the following:

\begin{lemma} %
  \label{thrm:K-has-rotations}
  There always exists a rotation DAG, call it $G$, such that $D(G)=\K$.
\end{lemma}
\begin{proof}
  Fix some instance with preferences and priorities $Q$.
  The main result of \cite{FaenzaZ22} is that there exists an altered instance $\widetilde Q$ such that $\Legal$ equals the set of stable matchings under $\widetilde Q$.
  Thus, lattice $\Legal$ under $Q$ has a predecessor DAG representation $(\underline{\mu}, R',P')$, where $\underline{\mu}$ the the result of running $\EADAM$ while interchanging the role of applicants and institutions, which is equal to the representation of the stable lattice under $\widetilde Q$.

  To transform this rotation DAG into an $(\IPDA,  R, P)$ which represents $\K$, it suffices to simply define $R$ by removing from $R'$ the set of rotations corresponding to $\IPDA$ (and setting $P$ to remove all elements of $P'$ containing removed elements of $R'$).
  After doing this operation, it is straightforward to verify that $D(\IPDA,R,P)$ contains exactly those elements of $D(\underline{\mu}, R',P')$ that weakly dominate $\IPDA$, as desired.
\end{proof}

\subsection{Structural Results for $\PN$ and the Good Chain Property}
\label{sec:structural-properties}

We now begin to reason about the structure of $\PN \subseteq \K$.
The main structural result we establish here is that, while many matchings in $\K$ may not be in $\PN$, there always exists a chain $\IPDA = \mu_1 < \mu_2 < \ldots < \mu_N = \EADAM$ in $\PN$, such that each $\mu_{i+1}$ differs from $\mu_i$ by a rotation in $\K$.
To prove this, we adapt and prove additional properties of several notions of matching chains from prior work, particularly \cite{TangY14,FaenzaZ22,Reny22}.

\subsubsection{Properties of \texorpdfstring{$\PN$}{PN}}

To begin, we recall one of the integral tools \cite{Reny22} uses to prove his results for $\PN$ matchings. 
This is the notion of a \emph{(Reny-)feasible chain}, which is defined as follows:

\begin{definition}[\cite{Reny22}]
    A chain $\mu_1 < \mu_2 < \ldots < \mu_N$ is \emph{Reny-feasible} if $\mu_1$ is stable, $\mu_N$ is Pareto-efficient, and for each $i$, $\mu_i$ does not violate the priority of any applicant who is Pareto-improvable in $\mu_{i-1}$.
\end{definition}

While \cite{Reny22} reasons about Reny-feasible sequences in quite some detail, he does not prove the following fundamental fact about Reny-feasible sequences: every matching in a Reny-feasible sequence is always priority-neutral. 
To prove this fact, we extend the techniques from \cite[Lemmas 3--6]{Reny22}.

First, we recall one simple lemma from \cite{Reny22}:

\begin{lemma}[\cite{EhlersM20}, Lemmas 4--6; \cite{Reny22}, Lemma 1]
  \label{lem:max-min-lattice-closure}
  Let $\mu$ and $\nu$ be such that for all applicants $d$: 
  $d$'s priority at $\mu(d)$ is not violated in $\nu$, 
  and $d$'s priority at $\nu(d)$ is not violated in $\mu$.
  Then, $\mu \vee \nu$ and $\mu \wedge \nu$ are matchings. 
\end{lemma}

We now prove our core result regarding Reny-feasible chains:

\begin{proposition}
    \label{thrm:Reny-feasible-are-all-pn}
  Let $\mu_1 < \mu_2 < \ldots < \mu_N$ be any Reny-feasible sequence.
  Then for each $n \in \{1,\ldots, N\}$, the matching $\mu_n$ is priority-neutral.
\end{proposition}
\begin{proof}
  Call $\xi$ an \emph{improver-only violator} of $\mu$ if, for all $d \in \vio(\xi)$, we have $\xi(d)\succeq_d \mu(d)$. 
  Note that a priority-correcting adjustment of $\mu$ is simply an improver-only violator, which also improves the match of some applicant whose priority is violated by $\mu$.
  We proceed by considering any $n$ and any improver-only violator $\xi$ of $\mu_n$, and showing that $\xi$ cannot improve the match of any applicant whose priority is violated by $\mu_n$.

  Fix some $n \in \{1,\ldots, N\}$, and suppose that $\xi$ is an improver-only violator of $\mu_n$.
  Using induction on $n'$, we will prove the following claim: 
  for all $n' \in \{1,\ldots, n\}$, we have that for all applicants $d$,\footnote{
    \cite[Lemmas 4 and 5]{Reny22} perform related inductive arguments, but get weaker conclusions that are only sufficient to conclude that the final matching $\mu_N$ in the sequence is priority-neutral.
  } 
  \[ d \in \vio(\mu_{n'}) \qquad \implies \qquad \mu_{n'}(d)\succeq_d \xi(d). \]
  The base case of $n' = 1$ follows from the fact that $\mu_1$ is stable by the definition of a Reny-feasible sequence, so $\vio(\mu_1) = \emptyset$.

  Now assume by induction that the claim holds for some $n'$. 
  We will prove the claim for $n'+1$.

  To begin, we will prove using \autoref{lem:max-min-lattice-closure} that $\mu_{n'}\vee \xi$ is a matching.
  To see this, observe that by induction, we know that for all $d \in \vio(\mu_{n'})$, we have $\mu_{n'}(d)\succeq_d \xi(d)$.
  Thus, $\mu_{n'}$ does not violate any $d$'s priority at $\xi(d)$.
  On the other hand, by the definition of an improver-only violator of $\mu_{n}$, we know that for all $d \in \vio(\xi)$, we have $\xi(d)\succeq_d \mu_n(d) \succeq_d \mu_{n'}(d)$.
  Thus, $\xi$ does not violate any $d$'s priority at $\mu_{n'}(d)$.
  Thus, \autoref{lem:max-min-lattice-closure} implies that $\mu_{n'}\vee \xi$ is a matching.

  Now, consider any applicant $d \in \vio(\mu_{n'+1})$. 
  By the definition of a Reny-feasible sequence, $d$ is not Pareto-improvable in $\mu_{n}'$.
  Since $\mu_{n'+1}\ge\mu_{n'}$, also by the definition of a Reny-feasible sequence, we thus have $\mu_{n'+1}(d)=\mu_{n'}(d)$.
  Thus, since $\mu_{n'}\vee\xi \ge \mu_{n'}$, we must have $[\mu_{n'} \vee \xi ](d) = \mu_{n'}(d)$, and therefore $\mu_{n'+1}(d)=\mu_{n'}(d) \succeq_d \xi(d)$.
  This finishes the proof of the inductive step, and proves the claim for all $n' \le n$.

  This result directly implies that $\mu_n$ must be priority-neutral.
  To see this, we apply the claim with $n' = n$.
  In detail, consider any $\xi$ which is a improver-only violator of $\mu_n$.
  By the claim, we know that for all $d \in \vio(\mu_n)$, we have $\mu_n(d) \succeq_d \xi(d)$. 
  Thus, $\xi$ is not a priority-correcting adjustment of $\mu_n$.
  Thus, in fact $\mu_n$ has no priority-correcting adjustments, and in fact $\mu_n$ must be priority-neutral.
\end{proof}

We will also need the following result, which says that when looking for priority correcting adjustments, it is without loss of generality to look at matchings $\ge \IPDA$. We prove this result using a novel strengthening of the classical argument used to show that $\IPDA$ is the institution-optimal stable match, going all the way back to \cite{GaleS62}. 

\begin{lemma}
    \label{thrm:pca-dominates-ipda}
  Consider any matching $\nu \in \K$, and let $\xi$ be any priority-correcting adjustment of $\nu$. 
  Then $\xi \ge \IPDA$.

  Moreover, consider any $\nu \ge \IPDA$, and suppose $\xi$ is a matching such that for all $d \in \vio(\xi)$, we have $\xi(d)\succeq_d \nu(d)$.
  Then $\xi \ge \IPDA$.
\end{lemma}
\begin{proof}
  Recall that for each $\nu \in \K$, we have $\nu \ge \IPDA$.
  As in the prove of \autoref{thrm:Reny-feasible-are-all-pn}, call $\xi$ an improver-only violator of $\mu$ if $d \in \vio(\xi)$ implies $\xi(d)\succeq_d\mu(d)$,
  and that priority-correcting adjustments of $\nu$ are just improver-only violators with additional requirements.
  Thus, the first paragraph of the lemma is actually a special case of the second paragraph.

  We now prove the second paragraph of the lemma.
  To begin, observe that if $\xi$ satisfies the above condition for any $\nu \ge \IPDA$, then it also satisfies the same condition for $\nu = \IPDA$ (since for each $d \in \vio(\xi)$, we have $\xi(d)\succeq_d \nu(d) \succeq_d \IPDA(d)$).
  Thus, it suffices to prove the lemma for the case where $\nu = \IPDA$.
  We prove the contrapositive, i.e. that if $\xi \ngeq \IPDA$, then there must exist a $d \in \vio(\xi)$ such that $\IPDA(d) \succ_d \xi(d)$.

  To this end, suppose that $\xi \ngeq \IPDA$. 
  This means there exists a $d$ such that $ \IPDA(d) \succ_d \xi(d)$.
  Let $h = \IPDA(d)$. 
  If $d \succ_h \xi(h)$, then we are done because $d \in \vio(\xi)$ and $\IPDA(d) \succ_d \xi(d)$.
  So suppose going forward that $\IPDA(h) = d \prec_h \xi(h)$.\footnote{
    In what follows, we use a classical argument from \cite{GaleS62} which to proves that $\IPDA$ is the institution-optimal stable matching, and we strengthen the conclusion of the argument.
    Namely, we prove that if any institution $h$ is matched to an applicant they prefer to $\IPDA(h)$, then not only is it the case that some applicant's priority is violated, but the priority of some applicant \emph{who prefers $\IPDA$} is violated.
  }

  Now, consider some sequence of proposals and rejections made during the execution of the $\IPDA$ algorithm.
  Since $\xi(h) \succ_h \IPDA(h)$, there exists at least one $h$ such that $\xi(h)$ rejects $h$ during $\IPDA$.
  Consider the \emph{first} such rejection; say, applicant $d_1$ rejects $h_1 = \xi(d_1)$ in favor of $h_2 \ne \xi(d_1)$;
  note that by the way the $\IPDA$ algorithm works, we have $h_2 \succ_{d_1} h_1$.
  Since this is the \emph{first} such rejection, we know that $h_2$ has not yet been rejected by $\xi(h_2)$, so it follows that $d_1 \succ_{h_2} \xi(h_2)$.
  Combining these last two facts, we see that $d_1 \in \vio(\xi)$.

  Moreover, by the way the $\IPDA$ algorithm works, $d_1$'s match can only improve through the remaining run of $\IPDA$; thus, $\IPDA(d_1) \succeq_{d_1} h_2 \succ_{d_1} h_1 = \xi(d_1)$.
  Thus, we get both that $d_1 \in \vio(\xi)$, and that $\IPDA(d_1) \succ_{d_1} \xi(d_1)$.
  This finishes the proof.
\end{proof}

\subsubsection{Properties of \texorpdfstring{$\Legal$}{L}}

We now introduce another type of chain in order to additionally reason about $\Legal$ and $\K$ while bringing in the techniques of \cite{TangY14}.
This chain modifies the Kesten-Tang-Yu chain (as defined in \autoref{def:Kesten-Tang-Yu-Alg}) by, informally speaking, progressing the chain with ``as small of differences as possible.''
To formalize this, we recall that in any lattice, an element $x$ \emph{covers} $y$ if $x \ge y$, and for any $x \ge z \ge y$, we have either $x=z$ or $y=z$.
Covering elements in stable lattice correspond to two matchings differing by rotations in the rotation poset.
In other words, $\mu_2$ covers $\mu_1$ in a stable lattice if and only if there is a rotation $\alpha$ exposed in $\mu_1$ such that $\alpha\mu_1 = \mu_2$.

\begin{definition}
    \label{def:gradual-KTY}
  We say that a chain of matchings is a \emph{gradual Kesten-Tang-Yu chain} if it can be constructed as follows:
  Set $P_0 = P$ to be the full preference (and priority) profile, and let $\mu_0 = \IPDA(P_0)$.
  Now, for $i\ge 1$, as long as $\mu_{i-1} \ne \APDA(P_{i-1})$, set $P_i = P_{i-1}$, and set $\mu_i$ equal to any matching which covers $\mu_{i-1}$ in the stable matching lattice under preferences $P_{i-1}$.
  If $\mu_{i-1} = \APDA(P_{i-1})$, yet $\mu_{i-1}$ is not Pareto-optimal, then modify $P_i$ as in \autoref{def:Kesten-Tang-Yu-Alg} (i.e., let $P_i$ denote the following operation on $P_{i-1}$: for each $d$ matched to an under-demanded institution in $\mu_{i-1}$, remove all institutions in their list strictly before $\mu_{i-1}(d)$).
  One can show $\mu_{i-1}$ is stable under preferences $P_i$ as well;
  we set $\mu_i$ to any matching stable under $P_i$ which covers $\mu_{i-1}$ in the stable lattice under $P_i$.
  The chain terminates with $\mu_N$ once $\mu_N$ is Pareto-optimal.
\end{definition}

By \cite[Lemma 2]{Reny22}, or by the techniques of \cite{TangY14}, we know that if $\mu_i$ in the gradual Kesten-Tang-Yu chain is not Pareto optimal, then some applicants must be matched to under-demanded institutions.\footnote{
  In more detail, \cite[Lemma 2]{Reny22} shows that if an applicant $d$ is Pareto-improvable at $\mu_i = \APDA(P_i)$, then $d$ cannot be matched to an under-demanded in $\APDA(P_i)$.
  Thus, whenever $\mu_i = \APDA(P_i)$ but $\mu_i$ is not Pareto-optimal, we know that $P_{i+1} \ne P_i$.
}
Conversely, if the matching is Pareto-optimal, then it must equal $\APDA$ under the current preferences (and under the original preferences).
Hence there always exists at least one gradual Kesten-Tang-Yu sequence. 

Additionally, we observe that any gradual Kesten-Tang-Yu sequence is also Reny-feasible.
To see why, observe that, by construction, a gradual Kesten-Tang-Yu sequence will only violate the priority of applicants matched to under-demanded institutions (and these applicants must be Pareto-unimprovable by \cite[Lemma 2]{Reny22}).
Moreover, the first matching $\mu_0$ of a gradual Kesten-Tang-Yu chain equals $\IPDA$, and the last matching $\mu_N$ is Pareto-efficient and equals $\EADAM$.
Thus: 
\begin{observation}
  \label{thrm:gradual-KTY-is-feasible}
  Any gradual Kesten-Tang-Yu sequence is Reny-feasible. 
\end{observation}

Hence, by \autoref{thrm:Reny-feasible-are-all-pn}, all elements of any gradual Kesten-Tang-Yu sequence are priority-neutral, and by \autoref{thrm:pn-in-l-l-lattice}, also legal.

We now use results from \cite{FaenzaZ22} and \cite{TangY14} to establish that any gradual Kesten-Tang-Yu sequence consists of covering relationships in the legal lattice; i.e., any gradual Kesten-Tang-Yu sequence is a chain in the legal lattice which starts with $\IPDA$, and is maximal among such chains.
This constitutes our core result regarding $\Legal$.
\begin{proposition}
    \label{thrm:gradual-KTY-covers-in-legal}
  Let $\mu_1 < \mu_2 < \ldots < \mu_N$ be any gradual Kesten-Tang-Yu sequence.
  Then, for each $n = 1,\ldots,N-1$, we have that $\mu_n$ is covered by $\mu_{n+1}$ in $\Legal$. 
\end{proposition}
\begin{proof}
    Fix some profile $P$ of preferences (and priorities).
    To begin, we recall some concepts from \cite{FaenzaZ22}.
    If applicant $d$ is never matched to institution $h$ in any matching in $\Legal$, then \cite{FaenzaZ22} call the pair $(d, h)$ an \emph{illegal pair}.
    Let $\widetilde{P}$ denotes the altered preferences where, for all $d$, we modify $d$'s preference list by removing all $h$ such that $(d,h)$ is an illegal pair.
    \cite[Lemma 3.5]{FaenzaZ22} shows that the legal lattice on $P$ equals the stable lattice on $\widetilde{P}$ (and moreover that illegal pairs can always be removed from applicants preference lists without changing the set of legal matchings).
    Our proof proceeds by showing that $\widetilde{P}$ can be constructed from the execution of a gradual Kesten-Tang-Yu algorithm as in \autoref{def:gradual-KTY} (along with additional steps after the sequence) in a way such that the successive elements of the gradual Kesten-Tang-Yu sequence correspond covering relationships in the stable lattice under $\widetilde{P}$.

    Now, consider $\mu_0 = \IPDA(P)$. 
    We know from the main result of \cite{TangY14} (or \cite[Lemma 2]{Reny22}) that in \autoref{def:gradual-KTY}, every $d$ matched to an under-demanded institution must be matched as in $\EADAM$.
    Thus, for each $h'$ above such a $d$'s current match in $\mu_i$ is such that $(d, h')$ is an illegal pair.
    Thus, for each step of \autoref{def:gradual-KTY}, for each $d$ matched to an under-demanded institution, all of the institutions preferred to $d$'s current match $h$ can be removed from $d$'s list without changing the set of legal matchings.
    Let $P_{\mathsf{top}}$ denote the result of applying this operation successively for all steps of \autoref{def:gradual-KTY}.
    Note that there are still additional illegal edges that remain on applicants' lists in $P_{\mathsf{top}}$; 
    since we know all matchings $\mu_i$ are legal (by \autoref{thrm:Reny-feasible-are-all-pn} and the fact that $\PN\subseteq\Legal$),
    these remaining illegal pairs are some set of pairs $(d,h)$ such that $\mu_i(d) \prec_d h \prec_d \mu_{i+1}(d)$ for some $i$ (or such that $h \prec_d \mu_1 = \IPDA$).
    However, we can see from the definition of rotations under some set of preferences %
    that deleting any such $h$ from $d$'s list will never change whether a given rotation $\alpha$ is an actual rotation under the preferences: %
    a rotation which was exposed in a matching $\mu$ before any further removal of such illegal pairs must remail exposed in $\mu$ after we remove these further illegal pairs, directly by the definition of a rotation exposed in $\mu$.
    (Note, however, that this operation may change the predecessor relationships; in particularly, one can show that it removes all type 2 predecessors.)
    Thus, when all illegal edges have been removed, and thus by \cite[Lemma 3.6]{FaenzaZ22} the legal set equals the stable set, we will still have that all rotations along the gradual Kesten-Tang-Yu sequence also constitute rotations in the legal lattice.
    Since one matching covers another in a stable matching lattice if and only if they differ by rotations, the result follows.
\end{proof}

\subsubsection{Good Chain Property}

We now prove our main structural result for $\PN$, which we call the ``good chain property''.
This result will combine \autoref{thrm:Reny-feasible-are-all-pn} and \autoref{thrm:gradual-KTY-covers-in-legal} in order to additionally show that there is some maximal chain in $\K$ such that every element in this maximal chain is also in $\PN$.

\begin{theorem}[Good Chain Property of $\PN$]
  \label{lem:struct-for-pn-max-chain}
  Let $G$ be the rotation DAG representation of $\K$.
  Then, there exists an ordering $[ \alpha_1, \ldots, \allowbreak \alpha_N ]$ of the rotations in $G$ such that for each $k \in \{1, \ldots, N\}$, we have $\match\big(\{\alpha_1,\ldots,\alpha_k\}) \in \PN$.
\end{theorem}
\begin{proof}
  Let $G = (\IPDA, R, P)$ be the rotation DAG representation of $\K$.
  Consider any gradual Kesten-Tang-Yu sequence $\IPDA = \mu_1 < \mu_2 < \ldots < \mu_N = \EADAM$.
  Since each $\mu_i \in \K$, they correspond to a set $X_i \subseteq R$ with $\match(X_i) = \mu_i$; also note that $X_1 = \emptyset$.
  Moreover, by \autoref{thrm:gradual-KTY-covers-in-legal}, each $\mu_i$ covers $\mu_{i-1}$ in $\Legal$, and hence in $\K$.
  Thus, for each $i = 1,\ldots,N$, we have $X_i = X_{i-1} \cup \{\alpha_i\}$ for some $\alpha_i \in R$ (since covering relationships in $D(G)$ always correspond to adding a single rotation).
  Additionally, by \autoref{thrm:gradual-KTY-is-feasible} and \autoref{thrm:Reny-feasible-are-all-pn}, each $\mu_i \in \PN$. 
  So, with respect to the ordering $[ \alpha_1, \alpha_2, \ldots, \alpha_N ]$ constructed as above, we know that for each $k \in \{1, \ldots, N\}$, we have $\alpha_k\alpha_{k-1}\ldots\alpha_1 \IPDA \in \PN$.
  This proves the theorem.
\end{proof}

\subsection{A Simpler Sufficient Representation of \texorpdfstring{$\PN$}{PN}: Fanout Lattices}
\label{sec:simpler-sufficient}

We now define an abstract class of lattice representations which generalizes the predecessor DAG representation of distributive lattices.
We will call this representation ``fanout multigraphs'', and we prove that the induced ``fanout lattices''
are sufficient to capture every possible $\PN$ lattice.
They can also express an analogue of the good chain property established in \autoref{lem:struct-for-pn-max-chain}.

\begin{definition}[Fanout multigraph of fanout lattices]
    \label{def:fanout-set}
    For some set of applicants $\Appls$ and institutions $\Insts$, a \emph{fanout multigraph} is a tuple $G = (\mu, R, E)$, defined as follows.
    First, $\mu$ is some base matching.
    Second, $R$ is some set of rotations $R = \{ \rho_1, \ldots, \rho_N \}$,
    equipped with a fixed ordering $[\rho_1,\ldots,\rho_N]$.
    Third, $E$ is a set of \emph{fanout} (multi)edges over $R$. 
    Each such fanout $e \in E$ is a pair $e = (r_0, \{r_1, r_2, \ldots, r_k\})$ for some $k\ge 1$ and $r_i\in R$ for $i\in\{0,1,\ldots,k\}$.

    Consider some subset $X\subseteq R$ and some fanout $e= (\rho_0, \{\rho_1, \rho_2, \ldots, \rho_k\})$ over $R$.
    We say that $e$ \emph{corrects} $X$ if $\rho_0 \in X$, but $X \cap \{\rho_1, \rho_2, \ldots, \rho_k\} = \emptyset$.\footnote{
        As we will see below, this terminology is inspired by the fact that if $\match(X) = \mu$ and $e$ corrects $X$, then $e$ corresponds to a priority-correcting adjustment of $\mu$.
    }

    We say that $G$ with order $[\rho_1,\ldots,\rho_N]$ over rotations has the \emph{good chain property} if for all $e = (\rho_j , T)\in E$, there exists an $i < j$ such that $\rho_i \in T$. (This means that for each $j=1,\ldots,N$, there is no fanout $e$ which corrects $\{ \rho_1, \rho_2, \ldots, \rho_j \}$.)

    Define the \emph{(abstract) fanout lattice}, denoted $F_{\mathsf{abs}}(G)$, to be the collection of subsets of $R$, ordered by set inclusion, defined as follows:
    \[ F_{\mathsf{abs}}(G) = \left\{ X \subseteq R\ \big|\ \text{No $e\in E$ corrects $X$} \right\}. \]

    Define the \emph{(matching) fanout lattice}, denoted $F(G)$, as follows.
    (As in \autoref{def:rotation-dag}, we require that for every $X =  \{ \rho_{i(1)}, \ldots \rho_{i(k)} \} \in F_{\mathsf{abs}}(G)$ where $i(1) < \ldots < i(k)$, we have $\rho_{i(j+1)}$ valid at $\rho_{i(j)}\ldots\rho_{i(1)}\mu$ for each $j=0,\ldots,k-1$.)
    We set
    \[ F(G) = \left\{ \match(X)\ \big|\ \text{$X \subseteq R$, and no $e\in E$ corrects $X$} \right\}. \qedhere \]
\end{definition}

We now show that the fanout lattice is (in fact) a lattice, with join operations given by $\cup$ (analogous to the way join operations in $\PN$ are given by coordinatewise maxima).

\begin{proposition}
    \label{thrm:fanouts-give-lattice}
    For fanout multigraph $G$, the set $F = F_{\mathsf{abs}}(G)$ forms a lattice under the set-inclusion ordering.
    Moreover, for all $S, T \in F$, we have $S \vee_{F} T = S \cup T$ (though the same is not true for $\wedge_{F}$).
\end{proposition}
\begin{proof}
    To begin, observe that $\emptyset \in F$ and $R \in F$, so $F$ has a least and greatest element under $\subseteq$.
    
    Next, we show that $F$ has least upper bounds given by $\cup$.
    Consider any $S, T \in F$. 
    We claim that $S \cup T \in F$.
    To see this, observe consider any fanout $e = (r_0, \{r_1,\ldots,r_k\})$ such that $r_0 \in S\cup T$.
    Then, $r_0$ is also in one of $S$ or $T$.
    Suppose without loss of generality that $r_0 \in S$.
    Since $S \in F$, we know $e$ does not correct $S$, thus, there is some $j \in \{1,\ldots,k\}$ such that $r_j \in S$. 
    Thus we also have $r_j \in S \cup T$, so $e$ does not correct $S\cup T$ either.
    Since this applies for all fanouts in $E$, we know that $S\cup T \in F$ as well.
    This shows $S$ and $T$ have a least upper bound in $F$, and that indeed $S \vee_F T = S\cup T$.

    Finally, we show that $F$ has greatest lower bounds.
    To see this, consider any $S, T \in F$. 
    Let $\mathsf{Under}_{S,T} = \{ U \in F \ \mid\ \text{$U \subseteq S$ and $U \subseteq T$} \} \subseteq F$ denote the subset of $F$ below both $S$ and $T$.
    Since $\emptyset \in F$, we have $\emptyset \in \mathsf{Under}_{S,T}$, and thus $\mathsf{Under}_{S,T} \ne \emptyset$.
    Thus,  the union of all the sets in $\mathsf{Under}_{S,T}$ is well-defined; call this union $\overline{U} = \bigcup_{U \in \mathsf{Under}_{S,T}} U$.
    Since we just showed $F$ is closed under unions, we have $\overline{U}\in F$.
    We claim that $\overline{U}$ is the greatest lower bound of $S$ and $T$ in $F$.
    To see this, note that $\overline{U}\subseteq S, T$ by construction, and moreover, for any $V \in F$ with $V \subseteq S, T$, we have $V \in \mathsf{Under}_{S,T}$ and thus $V \subseteq \overline{U}$.
    This shows $\overline U = S \wedge_{F} T$, as desired.
    Thus, $S$ and $T$ have a greatest lower bound in $F$, completing the proof.\footnote{
      After showing that $\PN$ is closed under componentwise maximums, \cite{Reny22} uses precisely this same logic to show that $\PN$ has greatest lower bounds.
    }
\end{proof}

For example, consider the abstract fanout multigraph with rotations $R = \{ \alpha, \beta, \gamma \}$ and a single fanout $e = (\beta, \{ \alpha, \gamma \} )$.
Then, the only subset $S$ of $R$ such that $e$ corrects $S$ is $\{\beta\}$.
Thus, $F_{\mathsf{abs}}(G)$ contains every subset of $R$ except for $\{\beta\}$.
See \autoref{fig:S7-uncolored-with-fanout} for an illustration.
The corresponding lattice is as in \autoref{fig:non-distributive-b}.

\begin{figure}[htbp]
    \begin{center}
    \begin{minipage}[c]{\textwidth}
        \centering
        \begin{minipage}[c]{0.45\textwidth}
            \includegraphics*[width=\textwidth]{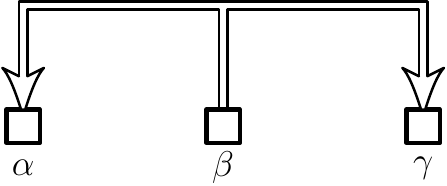}
        \end{minipage} 
    \end{minipage}
    \end{center}
\caption{Illustration of the simplest non-distributive fanout set; the corresponding fanout lattice is as in \autoref{fig:non-distributive-b}.}
\label{fig:S7-uncolored-with-fanout}
\end{figure}

It is not hard to see how fanout multigraphs generalize the order-theoretic properties of the rotation DAG (\autoref{def:rotation-dag}).
For, if $e = (r_0,r_1)$ is a predecessor edge, and $e' = (r_0, \{r_1\})$ is the analogous fanout, then $e$ violates a set $S$ of rotations if and only if $e'$ corrects $S$.
But fanouts additionally allow for other behavior, such as that exhibited in \autoref{fig:S7-uncolored-with-fanout} and \autoref{fig:non-distributive-b}
(and the additional more-involved examples below).

Before proceeding, we briefly discuss the good chain property for fanout multigraphs (which we show below follows via a direct consequence of \autoref{lem:struct-for-pn-max-chain}).
Without such a property, fanout multigraphs would have no notion of ``acyclicity'' (as required in the definition of a rotation DAG).
Indeed, a fanout multigraph without the good chain property could include two fanouts $(r_0, \{r_1\})$ and $(r_1, \{r_0\})$ in a length-two cycle, even though fanouts of this form have exactly the same semantics as predecessor relationships (which must be acyclic; \autoref{def:rotation-dag}).
Moreover, we conjecture that \emph{every} lattice has a representation as a fanout multigraph \emph{without} the good chain property, so we suspect that \emph{some} such notion of acyclicity must be assumed to achieve even a partial characterization which rules out some lattices.
See below in \autoref{sec:fanout-good-chain-non-example} for an example of how the good chain property can rule out some lattices as impossible.

\subsubsection{The Fanout Representation of \texorpdfstring{$\PN$}{PN}}
\label{sec:pn-embeds-in-fanout}

We now show that $\PN$ is always isomorphic to a fanout lattice;
moreover, using the results of \autoref{sec:structural-properties}, we show this fanout lattice always has the good chain property.

The intuition for why $\PN$ is a fanout lattice is that a fanout $e = (r_0, T)$ is able to exactly capture whether some fixed $\eta$ is a priority correcting adjustment of different matchings $\mu$. 
To give exposition into this, recall that a priority correcting adjustment of $\mu$ is some $\eta$ such that:
\begin{enumerate}[(1)]
    \item There is some $d_* \in \vio(\mu)$ with $\eta(d_*) \succ_{d_*} \mu(d_*)$.
    \item For all $d \in \vio(\eta)$, we have $\eta(d)\succeq_d \mu(d)$.
\end{enumerate}
Our goal is to set $e = (r_0,T)$ such that $e$ corrects $S\subseteq R$ if and only if $\eta$ is a priority-correcting adjustment of $\match(S)$.
Informally speaking, we do this by letting $r_0$ be the rotation that causes $d_*$'s priority to be violated, and by letting $T$ be all rotations $r$ such that, for $d \in \vio(\eta)$, rotation $r$ makes $d$ better than-off than $\eta(d)$.

We now state and prove our result formally:

\begin{proposition}
    \label{thrm:fanout-rep}
    For any profile of preferences and priorities, there exists a fanout multigraph $G$ with the good chain property such that $F(G) = \PN$.
\end{proposition}
\begin{proof}
    Recall from \autoref{thrm:K-has-rotations} that there exists a distributive lattice $\K \supseteq \PN$ which always has a predecessor DAG representation $(\IPDA, R, E_{\mathsf{DAG}})$
    for some set of predecessor edges $E_{\mathsf{DAG}}$.
    To prove the proposition, we construct a set of fanouts $E_{\mathsf{Fan}}$ such that for all $\mu\in\K$, we have $\mu\notin \PN$ if and only if there exist an $e \in E_{\mathsf{Fan}}$ such that $e$ corrects $\rots(\mu)$.

    To this end, fix some matching $\mu \in \K \setminus \PN$.
    By definition there must exist a priority-correcting adjustment $\eta$ of $\mu$.
    This means we have the following:
    \begin{enumerate}[(1)]
      \item There exists some $d_*$ such that $d_* \in \vio(\mu)$ and $\eta(d_*) \succ_{d_*} \mu(d_*)$. In particular, suppose $d_*$'s priority is violated at $h_*$, i.e. let $h_*$ be such that $h_* \succ_{d_*} \mu(d_*)$ and $d_* \succ_{h_*} \mu(h_*)$.
      \item For all $d \in \vio(\eta)$, we have $\eta(d)\succeq_d \mu(d)$. 
    \end{enumerate}
    Select any such an $\eta, d_*,$ and $h_*$.
    Now, we define a fanout $e(\mu) = (r_0, T)$ as follows.

    First, we define $r_0$. 
    Let $r_0$ denote the rotation moving $h_*$ to $\mu(h_*)$, i.e., the unique rotation $r_0$ such that, 
    for some matching $r_0 \widetilde\mu(h_*) = \mu(h_*)$ but $\widetilde\mu(h_*) \ne \mu(h_*)$ for any $\widetilde\mu \in \K$ where $r_0$ is valid. 
    There always exists such a rotation; to see why, note that if this was not the case, then $\IPDA(h_*) = \mu(h_*)$, but since $\mu \ge \IPDA$, we then must have $h_* \succ_{d_*} \IPDA(h_*)$, which contradicts the fact that $\IPDA$ is stable.

    Second, we define $T$.
    Let $T_1$ be the set of rotations moving $d_*$ to some institution at least as good for $d_*$ as $h_*$, 
    i.e., $T_1$ consists of all rotations $r$ such that all $\widetilde\mu \in \K$ with $r \in \rots(\widetilde\mu)$, we have $\widetilde\mu(d_*) \succeq_{d_*} h_*$. 
    Furthermore, let $T_2$ denote the set of rotations moving some applicant $d \in \vio(\eta)$ to some institution strictly better for $d$ than $\mu(d)$, i.e., 
    \[ T_2 = \{ r \ \mid\ 
        \text{for some $d \in \vio(\eta)$,
        we have $\widetilde\mu(d) \succ_d \eta(d)$ for all $\widetilde\mu$ with $r \in \rots(\widetilde\mu)$} \}.
    \]
    Then, we let $T = T_1\cup T_2$.\footnote{
        Note that $T_1$ or $T_2$ might be empty, but not both.
        This follows directly from \autoref{lem:fanout-rep-main-claim}, and the fact that $\match(R) = \EADAM \in \PN$, so there must be some $r \in T_1 \cup T_2$.
    }

    Observe that by construction, $r_0 \in \rots(\mu)$ and $(T_1 \cup T_2) \cap \rots(\mu) = \emptyset$, thus $e(\mu)$ corrects $\rots(\mu)$.
    This shows that by including fanout $e(\mu)$ in $E_{\mathsf{Fan}}$, we guarantee that $\mu \notin F(\IPDA,R,E_{\mathsf{Fan}})$; i.e., by using fanouts based on the above construction, we can always remove from our fanout representation all matchings which are not in $\PN$.
    To complete the proof, we must show that when we do not \emph{also} remove other matchings which should actually be in $\PN$.
    The following lemma guarantees this.
    \begin{lemma}
        \label{lem:fanout-rep-main-claim}
        If $\eta$ is defined as above, then
        for any matching $\nu \in \K$, if $e = (r_0, T_1\cup T_2)$ corrects $\rots(\nu)$, then $\eta$ is a priority-correcting adjustment of $\nu$.
    \end{lemma}
    To prove this lemma, suppose that $e$ corrects $\rots(\nu)$; this means that $r_0 \in \rots(\nu)$ and $T_i \cap \rots(\nu) = \emptyset$ for $i\in\{1,2\}$.
    The fact that $r_0 \in \rots(\nu)$ implies that $\nu(h_*) \prec_{h_*} d_*$.
    The fact that $T_1 \cap \rots(\nu) = \emptyset$ implies that 
        $\nu(d_*) \preceq_{d_*} \mu(d_*) \prec_{d_*} \eta(d_*)$.
    This, combined with the previous sentence, implies that $d_*$'s priority is violated in $\nu$, and that $\eta(d_*) \succeq_{d_*} \nu(d_*)$.
    The fact that $T_2 \cap \rots(\nu) = \emptyset$ implies that 
        for all $d \in \vio(\eta)$, we have $\nu(d) \preceq_{d} \eta(d)$.
    Thus, we see that $\eta$ is a priority-correcting adjustment of $\nu$.
    This concludes the proof of \autoref{lem:fanout-rep-main-claim}.

    We are now able to complete the proof of \autoref{thrm:fanout-rep}.
    Let $E_{\mathsf{Fan}} = \{ e(\mu) \ \mid \ \mu\in\K\setminus\PN \}$ (where for any such $\mu \in \K\setminus\PN$, we select any $\eta, d_*,$ and $h_*$ in order to define $e(\mu)$ as above), and let $G = (\IPDA, R, E_{\mathsf{fan}})$.
    Since $e(\mu)$ corrects $\rots(\mu)$ for all $\mu \in \K\setminus\PN$, we know that $M$ contains no elements of $\K\setminus\PN$.
    Moreover, 
    since \autoref{lem:fanout-rep-main-claim} shows that $e(\mu)$ corrects $\rots(\nu)$ only if $\nu$ has some priority-correcting adjustment,
    and since each $\mu \in \PN$ has no priority-correcting adjustments by definition, we see that $M$ contains all elements of $\PN$. 
    This shows that $F(G) = \PN$.

    Finally, we show that $G$ has the good chain property.
    But this follows directly from the fact that $F(G)=\PN$, and from \autoref{lem:struct-for-pn-max-chain}, since the chain constructed in that results can exactly tell us the order on $R$ needed.
    This concludes the proof.
\end{proof}

The above result shows why a fanout is able to capture the property of some fixed $\eta$ being a priority correcting adjustment of different matchings $\mu$.
When coupled with the simple proof in \autoref{thrm:fanouts-give-lattice} that fanout lattices are closed under set unions, this gives intuition for the perhaps mysterious-seeming property that $\PN$ is closed under coordinatewise maximums (but not minimums).\footnote{
    We note, however, that this approach would require much more work in order to achieve a self-contained proof that $\PN$ is closed under coordinatewise maxima.
    This is because we crucially use $\PN\subseteq \K$ to get the set of rotations representing (a superset of) $\PN$, and showing that $\PN\subseteq \K$ requires involved results from \cite{Reny22}.
}

\subsubsection{Examples of \texorpdfstring{$\PN$}{PN} Lattices and Fanout Representations}
\label{sec:fanout-examples}

In \autoref{fig:advanced-intersecting-fanouts} and \autoref{fig:hossy-wide-fanouts}, we give two more-involved examples of $\PN$ rotations and their fanout representations.
The figures show preferences and priorities which induces these lattices, and discuss how one can verify these examples.

\begin{figure}[htbp]
    
    \begin{minipage}[c]{\textwidth}
        \centering
        \begin{minipage}[c]{0.35\textwidth}
        \begin{center}
            \includegraphics*[height=1.2\textwidth]{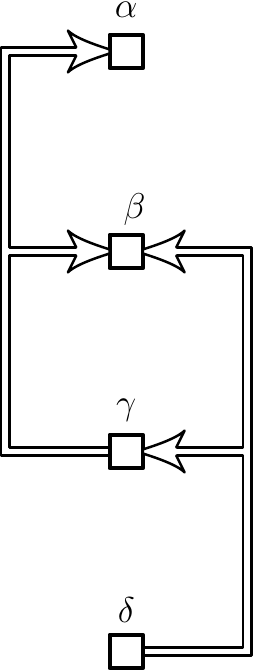}
            \subcaption{The fanout set $\big\{\ (\gamma, \{\alpha,\beta\}),\ \allowbreak  (\delta, \{\beta,\gamma\})\ \big\}$.}
        \end{center}
    \end{minipage} 
    \begin{minipage}[c]{0.55\textwidth}
        \begin{center}
            \includegraphics*[height=0.9\textwidth]{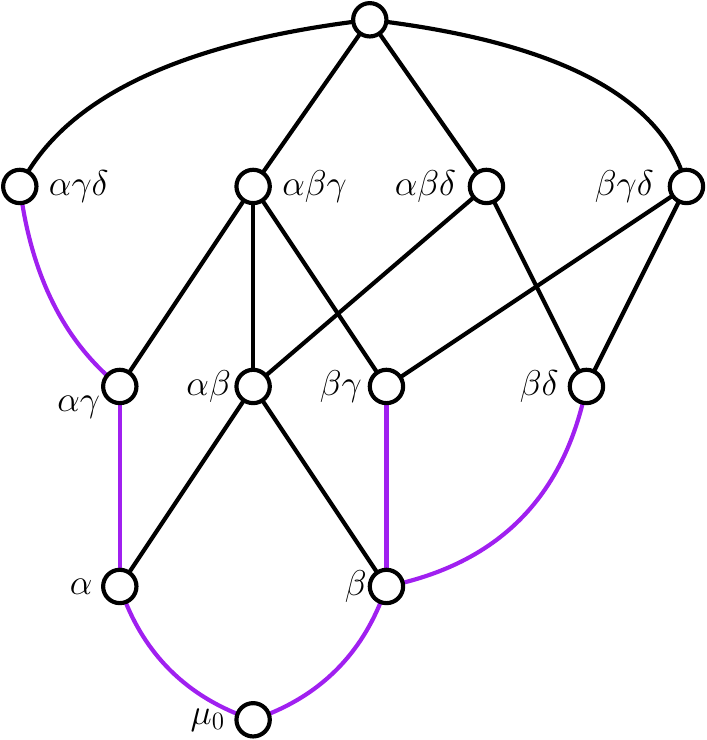}
            \subcaption{Corresponding lattice.}
        \end{center}
        \end{minipage} 
    \end{minipage}

    \vspace{0.3in}

    \begin{minipage}[t]{0.45\textwidth} \centering
        \begingroup
        \setlength{\tabcolsep}{3pt} %
        Institutions' Priorities:
        \\[0.05in]
        \begin{tabular}{cccccccccc}
            \toprule
            $h_A$ & $h_B$
            & $h_C$ & $h_D$
            & $h_E$ & $h_F$
            & $h_X$
            & $h_G$ & $h_H$
            & $h_Y$ \\
            \midrule
            $d_1$ & $d_2$
            & $d_3$ & $d_4$
            & $d_5$ & $d_6$
            & $d_u$
            & $d_7$ & $d_8$
            & $d_v$ \\
            $d_2$ & $d_1$
            & $d_4$ & $d_3$
            & $d_u$ & $d_5$
            & $\myvdots$
            & $d_v$ & $d_7$
            & $\myvdots$ \\
            $\myvdots$ & $\myvdots$ 
            & $\myvdots$ & $\myvdots$
            & $d_6$ & $\myvdots$ 
            & 
            & $d_8$ & $\myvdots$
            & \\
              &
            & &
            & $\myvdots$ &
            & 
            & $\myvdots$ & 
            &
        \end{tabular}
        \endgroup
    \end{minipage}
    \begin{minipage}[t]{0.45\textwidth} \centering
        \begingroup
        \setlength{\tabcolsep}{3pt} %
        Applicants' Preferences:
        \\[0.05in]
        \begin{tabular}{ccccccccccc}
            \toprule
            $d_1$ & $d_2$
            & $d_3$ & $d_4$
            & $d_5$ & $d_6$
            & $d_u$
            & $d_7$ & $d_8$
            & $d_v$ \\
            \midrule
            $h_B$ & $h_A$
            & $h_D$ & $h_C$
            & $h_F$ & $h_E$
            & $h_E$
            & $h_H$ & $h_G$
            & $h_G$ \\
            $h_X$ & $h_B$
            & $h_Y$ & $h_D$
            & $h_C$ & $h_F$
            & $h_X$
            & $h_E$ & $h_H$
            & $h_Y$ \\
            $h_A$ &  
            & $h_A$ & 
            & $h_E$ &  
            &
            & $h_G$ & 
            &  \\
             &
            & $h_C$ & 
            &  & 
            & 
            &  &
        \end{tabular}
        \endgroup
    \end{minipage} 
\begin{minipage}{\textwidth} \centering
    \subcaption{Preferences and priorities inducing this $\PN$ lattice.}
\end{minipage} 

    \caption{Illustration of a $\PN$ lattices via fanout sets.}
    \label{fig:advanced-intersecting-fanouts}
    \vspace{0.1in}
    {\footnotesize
    \textbf{Notes:} In $\IPDA$, every institution gets their top choice. The rotations are as follows: 
    $\alpha = [ (h_A,d_1), (h_B,d_2)]$, 
    $\beta = [ (h_C,d_3), (h_D,d_4)]$, 
    $\gamma = [ (h_E,d_5), (h_F,d_5)]$, 
    and
    $\delta = [ (h_G,d_7), (h_H,d_8)]$.
    The set of fanouts is
    $\big\{\ (\gamma, \{\alpha,\beta\}),\ \allowbreak (\delta, \{\beta,\gamma\})\ \big\}$.
    The join-irreducible elements of the lattice are shown with purple-shaded edges below them.

    \hspace{0.2in} To verify that this fanout set represents the $\PN$ lattice in this instance, one can take the following approach.
    First, one can check that $\IPDA < \alpha < \alpha\beta < \alpha\beta\gamma < \alpha\beta\gamma\delta$ is a gradual Kesten-Tang-Yu chain (\autoref{def:gradual-KTY}); thus, the set of rotations is indeed $\{\alpha,\beta,\gamma,\delta\}$.
    Second, using \autoref{thrm:pca-dominates-ipda}, we know that all priority-correcting adjustments Pareto improve on $\IPDA$;
    one can verify that that there are exactly two Pareto improvements to $\IPDA$ which are not elements of $\K$; namely, the two matchings 
    \begin{align*}
         \{(d_1, h_X), (d_u, h_E), (d_5, h_C), (d_3, h_A)\quad , \quad (d_2,h_B), (d_4,h_D), (d_6,h_F), (d_7,h_G), (d_8,h_H), (d_v,h_Y)\} \\
         \{(d_3, h_Y), (d_v, h_G), (d_7, h_E), (d_5, h_C) \quad , \quad (d_1,h_A), (d_2,h_B), (d_4,h_D), (d_6,h_F), (d_8,h_H), (d_u,h_X)\} \\
    \end{align*}

    \vspace{-0.2in}

    (where we visually divide participants who are matched as in $\IPDA$ with extra spacing for readability).
    One can check that these two matching correspond respectively to the fanouts $(\gamma, \{\alpha,\beta\})$ and $(\delta, \{\beta,\gamma\})$.
    \par}
\end{figure}

\begin{figure}[htbp]

    \begin{minipage}[c]{\textwidth}
        \centering
        \begin{minipage}[c]{0.35\textwidth}
        \begin{center}
            \includegraphics*[height=1.2\textwidth]{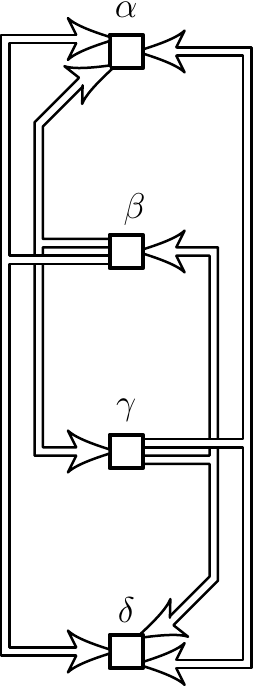}
            \subcaption{The fanout set.}
        \end{center}
    \end{minipage} 
    \begin{minipage}[c]{0.55\textwidth}
        \begin{center}
            \includegraphics*[height=0.9\textwidth]{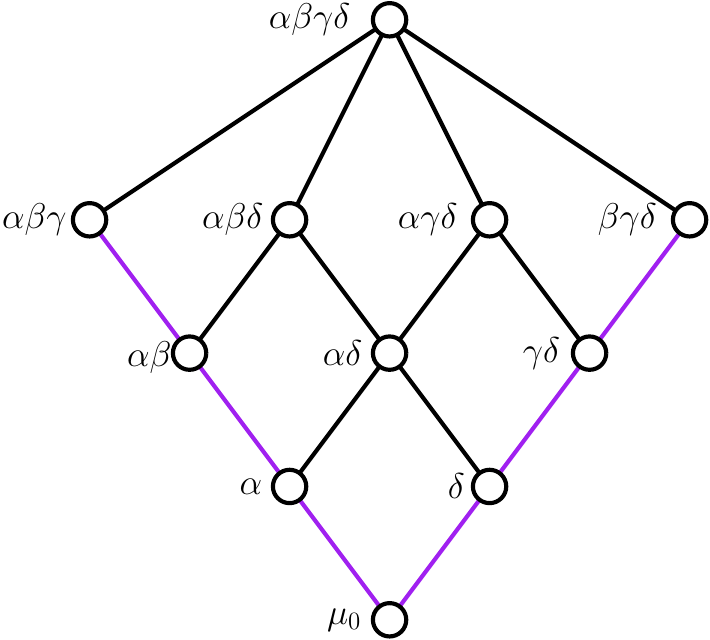}
            \subcaption{Corresponding lattice.}
        \end{center}
        \end{minipage} 
    \end{minipage}

    \vspace{0.3in}

    \begin{minipage}[t]{0.45\textwidth} \centering
        \begingroup
        \setlength{\tabcolsep}{3pt} %
        Institutions' Priorities:
        \\[0.05in]
        \begin{tabular}{cccccccccc}
            \toprule
            $h_A$ & $h_B$
            & $h_C$ & $h_D$
            & $h_E$ & $h_F$
            & $h_X$
            & $h_G$ & $h_H$
            & $h_Y$ \\
            \midrule
            $d_1$ & $d_2$
            & $d_3$ & $d_4$
            & $d_5$ & $d_6$
            & $d_u$
            & $d_7$ & $d_8$
            & $d_v$ \\
            $d_2$ & $d_1$
            & $d_u$ & $d_3$
            & $d_v$ & $d_5$
            & $\myvdots$
            & $d_8$ & $d_7$
            & $\myvdots$ \\
            $\myvdots$ & $\myvdots$ 
            & $d_4$ & $\myvdots$
            & $d_6$ & $\myvdots$ 
            & 
            & $\myvdots$ & $\myvdots$
            & \\
              &
            & $\myvdots$ &
            & $\myvdots$ &
            & 
            & & 
            &
        \end{tabular}
        \endgroup
    \end{minipage}
    \begin{minipage}[t]{0.45\textwidth} \centering
        \begingroup
        \setlength{\tabcolsep}{3pt} %
        Applicants' Preferences:
        \\[0.05in]
        \begin{tabular}{ccccccccccc}
            \toprule
            $d_1$ & $d_2$
            & $d_3$ & $d_4$
            & $d_u$
            & $d_5$ & $d_6$
            & $d_7$ & $d_8$
            & $d_v$ \\
            \midrule
            $h_B$ & $h_A$
            & $h_D$ & $h_C$
            & $h_C$
            & $h_F$ & $h_E$
            & $h_H$ & $h_G$
            & $h_C$ \\
            $h_X$ & $h_B$
            & $h_G$ & $h_D$
            & $h_E$
            & $h_A$ & $h_F$
            & $h_Y$ & $h_H$
            & $h_E$ \\
            $h_A$ &  
            & $h_C$ & 
            & $h_X$
            & $h_E$ &  
            & $h_G$ & 
            & $h_Y$ \\
        \end{tabular}
        \endgroup
    \end{minipage} 
\begin{minipage}{\textwidth} \centering
    \subcaption{Preferences and priorities inducing this $\PN$ lattice.}
\end{minipage} 

    \caption{Illustration of a $\PN$ lattices via fanout sets.}
    \label{fig:hossy-wide-fanouts}
    \vspace{0.1in}
    {\footnotesize
    \textbf{Notes:} 
    The set of fanouts is
    $\big\{\ (\beta, \{\alpha,\gamma\}),\ (\beta, \{\alpha,\delta\}),\ (\gamma, \{\beta,\delta\}),\ \allowbreak (\gamma, \{\alpha,\delta\})\ \big\}$.
    All other notation (i.e., the definition of the rotations $\alpha, \beta, \gamma, \delta$) is as in \autoref{fig:advanced-intersecting-fanouts}.

    \hspace{0.2in} To verify that the $\PN$ lattice is as described, one can take the same approach as in \autoref{fig:advanced-intersecting-fanouts}.
    The matchings which are not in $\K$ but which Pareto-improve on $\IPDA$ include:
    \begin{align*}
         \mu_1 & = 
         \{(d_1, h_X), (d_u, h_E), (d_5, h_A) \quad , \quad (d_2,h_B), (d_3,h_C), (d_4,h_D), (d_6,h_F), (d_7, h_G), (d_8,h_H), (d_v, h_Y)\}
        \\ \mu_2 & = 
         \{(d_3, h_G), (d_7, h_Y), (d_v, h_C) \quad , \quad (d_1,h_A), (d_2,h_B), (d_4,h_D), (d_u, h_X), (d_5, h_E), (d_6,h_F), (d_8,h_H)\}
        \\ \mu_3 & = 
         \{(d_1, h_X), (d_u, h_C), (d_3, h_G), (d_7, h_Y), (d_v,h_E), (d_5,h_A) \quad , \quad (d_2,h_B), (d_4,h_D), (d_6,h_F), (d_8,h_H)\} \\
    \end{align*}

    \vspace{-0.2in}
    (where we visually divide participants who are matched as in $\IPDA$ with extra spacing for readability).
    Pareto-improvements to $\IPDA$
    also include several ways to eliminate rotations on top of $\mu_1$ or $\mu_2$, and include the coordinatewise maximum of $\mu_1$ and $\mu_2$.
    Matching $\beta\mu_1$ corresponds to fanout $(\beta, \{\alpha, \gamma\})$; 
    matching $\gamma\mu_2$ corresponds to fanout $(\gamma, \{\beta,\delta\})$, 
    and matching $\mu_3$ corresponds to two different fanouts: $(\beta, \{\alpha,\delta\})$ and $(\gamma, \{\alpha,\delta\})$.
    \par}
\end{figure}

We also now discuss how these lattices have more-complicated features than would be possible for distributive lattices.
We phrase this discussion in terms of the join-irreducible elements of the lattice, i.e., those $x$ for which $a \vee b = x$ implies $a = x$ or $b = x$.
These elements are show in the figures with a purple-shaded edge below them.

To see that the lattice in \autoref{fig:advanced-intersecting-fanouts} is not distributive, recall by Birkhoff's representation theory \cite{Birkhoff37} that in a distributive lattices (under a predecessor DAG representation) each rotation can appear in exactly one join-irreducible element (namely, the join irreducible set containing $\alpha$ corresponds to the smallest downward-closed subset of the rotations which contains $\alpha$;
see \cite{echenique2023online}, or for a more general discussion in lattice theory, see e.g., \cite{gratzer2012general}).
In contrast, in the lattice in \autoref{fig:advanced-intersecting-fanouts}, both $\gamma$ and $\delta$ appear in multiple join-irreducible elements.
(A similar remark holds for $\beta$ and $\gamma$ in \autoref{fig:hossy-wide-fanouts}.)
Hence, these lattices break this key feature of distributive lattices in a strong way.

\subsubsection{Example of a \emph{Non}-\texorpdfstring{$\PN$}{PN} Lattice Due to the Good Chain Property}
\label{sec:fanout-good-chain-non-example}

We how show how the fanout representation of $\PN$, and the good chain property in particular, can prove that some lattices cannot ever arise as the $\PN$ lattice.

\autoref{fig:m3} displays a lattice, typically denoted by $M_3$, which is a famous simple example of a non-distributive lattice from the order theory literature \cite{birkhoff1940lattice, gratzer2012general}.

\begin{figure}[htbp]
    \begin{minipage}[c]{\textwidth}
      \centering
      \begin{minipage}[c]{0.45\textwidth}
      \centering
          \includegraphics*[height=0.6\textwidth]{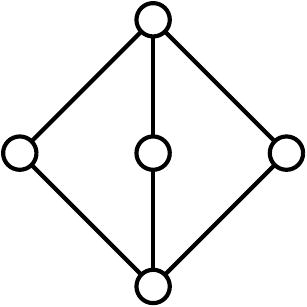}
          \subcaption{The lattice $M_3$.}
          \label{fig:m3}
      \end{minipage} 
      \begin{minipage}[c]{0.45\textwidth}
          \includegraphics*[height=0.6\textwidth]{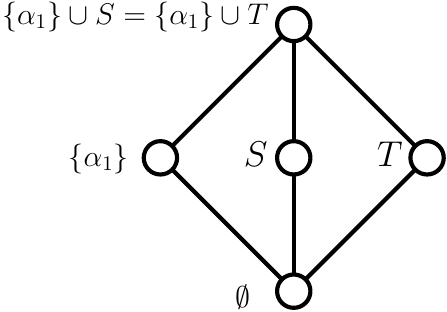}
          \subcaption{Key steps of the proof.}
          \label{fig:m3-proof}
      \end{minipage} 
    \end{minipage}
    \caption{$M_3$ lattice and key steps of the proof of \autoref{thrm:m3-impossibility}.}
    \label{fig:m3-overall}
\end{figure}

\begin{proposition}[$M_3$ is not a $\PN$ lattice]
  \label{thrm:m3-impossibility}
  There is no profile of priorities and preferences such that $\PN$ is isomorphic to the lattice $M_3$ shown in \autoref{fig:m3}.
\end{proposition}
\begin{proof}
  Suppose for contradiction there is some set of preferences by which $\PN$ is isomorphic to $M_3$.
  Let $G = (\IPDA, R, E)$ be the fanout multigraph representing $\PN$, as given by \autoref{thrm:fanout-rep}.
  Let $[ \alpha_1, \alpha_2, \ldots, \alpha_N ]$ be the ordering on $R$ as given by the good chain property.
  Observe that this means there is some element $\mu_1 \in \PN$ which corresponds to $\{\alpha_1\} \subseteq R$.

  Now, since we have assumed $\PN$ is isomorphic to $M_3$, there are two other matchings, call them $\mu_s \ne \mu_t$, which cover $\IPDA$ in $\PN$. 
  Suppose these matchings correspond to sets of rotations $S\subseteq R$ and $T \subseteq R$.
  Now, observe that since $\PN$ is closed under coordinatewise maxima (\autoref{thrm:reny-lattice}), we have that the join of two elements of $\PN$ corresponds to the union of the sets of rotations corresponding to the two elements.
  Thus, $\mu_1 \vee_{\PN} \mu_s$ corresponds to $\{\alpha_1\}\cup S$, and $\mu_1 \vee_{\PN} \mu_t$ corresponds to $\{\alpha_1\}\cup T$.

  However, by the definition of $M_3$, we know that $\mu_1 \vee_{\PN} \mu_s = \mu_1 \vee_{\PN} \mu_t$.
  Thus, we must have $\{\alpha\}\cup S = \{\alpha\}\cup T$.
  (For an illustration of this argument, see \autoref{fig:m3-proof}.)
  This implies that $S = T$.\footnote{
      This step of the proof would fail if we were not able to establish that at least one matching in $\PN$ differed from $\IPDA$ by only one rotation.
      Indeed, if we replaced $\{\alpha_1\}$ with a set of just two rotations $\{\alpha_1, \beta\}$, then we would not be able to conclude from $\{\alpha_1, \beta\}\cup S = \{\alpha_1, \beta\}\cup T$ that $S = T$.
      This shows how we use the good chain property.
  }
  But this contradicts the assumption that $S$ and $T$ correspond to distinct elements of the $\PN$ lattice.
  This concludes the proof.
\end{proof}

We make two additional remark regarding the lattice $M_3$.
First, the order theory literature calls a lattice \emph{modular} if $a \le b$ implies that for all $x$, we have $a \vee (x \wedge b) = (a \vee x) \wedge b$.
One can show that every distributive lattice is modular, and indeed modular lattices form a canonical well-studied class of lattices capturing many mathematical insights \cite{birkhoff1940lattice,gratzer2012general}.
The lattice $M_3$ in \autoref{fig:m3} is a classic example of a modular but not distributive lattice.
Thus, \autoref{thrm:m3-impossibility} shows that not all modular lattices arise as priority-neutral lattices.
On the other hand, the $\PN$ lattice in \autoref{fig:non-distributive-b} is not modular, since $\alpha \vee (\beta\gamma \wedge_{\PN} \alpha\beta) = \alpha \ne \alpha\beta = (\alpha \vee \beta\gamma) \wedge_{\PN} \alpha\beta$.
Thus, not all priority-neutral lattices are modular.
Thus, the class of modular lattices---a canonical generalization of distributive lattices---is unrelated (in the set-inclusion sense) to the class of lattices which arise as $\PN$.

Second, we observe that without using the good chain property, there is a fanout representation of $M_3$, shown in \autoref{fig:m3-ignoring-good-chain}.

\begin{figure}[htbp]
    \begin{minipage}[c]{\textwidth}
        \centering
        \begin{minipage}[c]{0.45\textwidth}
        \begin{center}
            \includegraphics*[height=\textwidth]{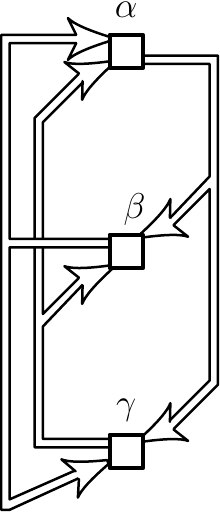}
            \subcaption{A fanout set.}
        \end{center}
    \end{minipage} 
    \begin{minipage}[c]{0.45\textwidth}
        \begin{center}
            \includegraphics*[height=0.8\textwidth]{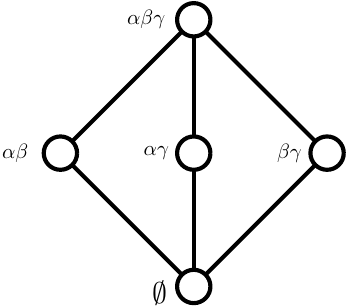}
            \subcaption{Corresponding lattice.}
        \end{center}
        \end{minipage} 
    \end{minipage}

    \caption{Illustration of a fanout representation of $M_3$ without the good chain property.}
    \label{fig:m3-ignoring-good-chain}
\end{figure}

\subsubsection{Remaining Challenges for a Tight Characterization}
\label{sec:remaining-challenges}

\paragraph{Obstacle: Acyclicity in Fanouts.}
The good chain property can be understood as a weak acyclicity property for fanout sets of $\PN$ lattices: for each fanout $(r, T)$, at least one of the elements of $T$ must come before $r$ in some global ordering.
Stronger notions of acyclicity do not hold for fanout sets.
For example, in \autoref{fig:hossy-wide-fanouts}, there are fanouts of the form $(\beta, S_\beta)$ and $(\gamma, S_\gamma)$ where $\beta \in S_\gamma$ and $\gamma \in S_\beta$.

While the good chain property enforces some structure on fanout sets, and is indeed strong enough to rule out some lattices (such as $M_3$ as in \autoref{thrm:m3-impossibility}), this property alone is not enough to characterize $\PN$ lattices.
For example, consider the lattice shown in \autoref{fig:n5}.
This lattice is known as $N_5$ in the order theory literature, and is another famous example of a non-distributive lattice.
The figure shows a fanout set representing $N_5$, and this fanout set has the good chain property. 
However, as we will show in \autoref{thrm:n5-not-priority-neutral} below once we have completed our full characterization, $N_5$ does not arise as a $\PN$ lattice.

\begin{figure}[htbp]
    \begin{minipage}[c]{\textwidth}
        \centering
        \begin{minipage}[c]{0.45\textwidth}
        \begin{center}
            \includegraphics*[height=\textwidth]{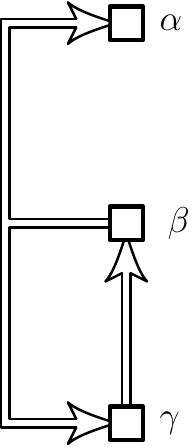}
            \subcaption{A fanout set.}
    \label{fig:n5-fanout}
        \end{center}
    \end{minipage} 
    \begin{minipage}[c]{0.45\textwidth}
        \begin{center}
            \includegraphics*[height=\textwidth]{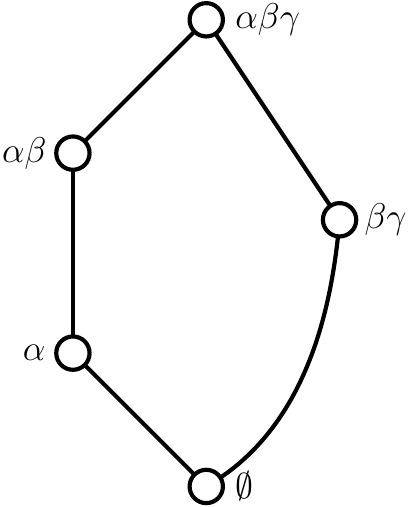}
            \subcaption{Corresponding lattice.}
    \label{fig:n5-lattice}
        \end{center}
        \end{minipage} 
    \end{minipage}

    \caption{Illustration of a \emph{non}-$\PN$ lattices, which still has a fanout representation with the good chain property.}
    \label{fig:n5}
\end{figure}

\paragraph{Idea: Bottom-Up Respresentation Approach.}
Given the above discussion, we suspect that exactly characterizing $\PN$ in terms of fanout sets would be challenging.
Instead, we provide a different approach.
Where fanout sets take a ``top down approach'', and generalize the order-theoretic properties of a rotation DAG (\autoref{def:rotation-dag}),
our full characterization takes a ``bottom up approach'', and generalizes the matching-theoretic properties of the rotation DAG of a specific stable lattice (see \autoref{thrm:stable-rotation-dag}).
In other words, our full representation encodes all possible ``components'' of priority-correcting adjustments (\autoref{def:priority-neutral}), instead of requiring a separate piece of data for every priority-correcting adjustment (as in fanout multigraphs).
See below for the full construction.

\section{Main Result}
\label{sec:result}

In this section, we give our main result: a tight characterization of priority-neutral lattices.
First, in \autoref{sec:movement-graph}, we define our characterization and prove that it can represent all $\PN$ lattices.
Second, in \autoref{sec:all-movement-is-PN}, we show that all lattices with this representation arise as $\PN$ lattices.
Third, in \autoref{sec:discussion}, we conclude the paper with a brief discussion and an application: the first poly-time algorithm for checking whether a given matching is priority-neutral.

\subsection{The Movement Graph Representation of \texorpdfstring{$\PN$}{PN}}
\label{sec:movement-graph}

We now state our main construction, called the \emph{movement graph} and induced \emph{movement lattice}, which we show (along with the good chain property) characterizes $\PN$ lattices.

\begin{definition}[Movement graph of movement lattices]
    \label{def:movement-graph}
    For some set of applicants $\Appls$ and institutions $\Insts$, a \emph{movement graph} is a tuple $G = (\mu, R, U, P, M)$, defined as follows.
    First, $\mu$ is some base matching.
    Second, $R$ is some set of rotations $R = \{ \rho_1, \ldots, \rho_N \}$,
    equipped with a fixed ordering $[\rho_1,\ldots,\rho_N]$.
    Third, $U \subseteq \Insts\times\Appls$ is a set we refer to as \emph{terminal pairs}.
    Fourth, $P$ is a set of predecessor relations on $R \cup U$, i.e., $P \subseteq (R\cup U)\times(R\cup U)$, such that if $(\rho_i, \rho_j) \in P$ for $\rho_i,\rho_j \in R$, then $i > j$.
    Fifth, $M$ is another subset of $(R\cup U)\times(R\cup U)$ which we call \emph{movement relations}.
    We think of $(x,y) \in P \cup M$ as the directed edges $x \to y$ of a graph on $R \cup U$, with edges labeled either ``predecessor'' or ``movement.''

    Now, for a given subset $S\subseteq R$, a \emph{PCA cycle} (named for priority-correcting adjustments; see \autoref{def:priority-neutral}) of $S$ is some sequence $C$ of the form
    \[  x_1 \xrightarrow{e_1} x_2 \xrightarrow{e_2} x_3 \xrightarrow{e_3}
    \ldots \xrightarrow{e_{K-1}} x_K \xrightarrow{e_K} x_1,
    \]
    where $x_i \in R\cup U$ for each $i\in[K]$ and there is a directed edge $e_i$ in $P\cup M$ from $x_i$ to $x_{i+1}$ for each $i \in [K]$ (cyclically), with the following two conditions:
    \begin{enumerate}[(1)]
        \item There exists a $i_* \in [K]$ such that $x_{i_*} \in U$, and moreover, there exists an $\alpha \in S$ such that $(\alpha, x_{i_*}) \in P$. 
        That is, there is some rotation $\alpha \in S$ such that terminal pair $x_{i_*}$ is a predecessor of $\alpha$. %
        (Note that $\alpha$ may or may not appear in $C$.)
        \item For each $i \in [K]$, if the edge $e_i$ with $x_i \xrightarrow{e_i} x_{i+1}$ is such that $e_i \in M$, then $x_i \notin S$.
        In other words, if $x_i \in R$ and moreover $x_i \in S$, then $e_i \in P$ is a predecessor edge.
    \end{enumerate}
    When this holds, we say that PCA cycle $C$ \emph{corrects} $S$.
    (As we will see below, if a PCA cycle $C$ corrects $S$, this will mean that $\match(S)$ has a priority-correcting adjustment, and is hence not in $\PN$.
    Moreover, the above (1) and (2) are the analogues of (1) and (2) in \autoref{def:priority-neutral}.) 

    Additionally, as in \autoref{def:rotation-dag}, we say that a predecessor relationship $p = (r_0, r_1)\in P \cap (R \times R)$ \emph{violates} set $S$ if $r_0 \in S$ but $r_1 \notin S$.
    That is, if $r_0$ and $r_1$ are both rotations, and $r_0$ but not $r_1$ are in $S$.

    In analogy with \autoref{def:fanout-set}, we say that $G$ has the \emph{good chain property} with respect to the ordering $[\alpha_1,\ldots,\alpha_N]$ over $R$ if, for each $i = \{ 0, 1, \ldots, N\}$, if we consider the set $S_i = \{ \alpha_1, \ldots, \alpha_i \}$, then no predecessor relationship violates $S_i$ and no PCA cycle corrects $S_i$. 

    Define the \emph{(abstract) movement lattice}, denoted $M_{\mathsf{abs}}(G)$, to be the collection of subsets of $R$, ordered by set inclusion, defined as follows:
    \begin{align*}
        M_{\mathsf{abs}}(G) = \big\{\ S \subseteq R  \ \mid \ 
          & \text{no predecessor relationship $p \in P$ violates $S$,} \\
          & \text{and no PCA cycle $C$ corrects $S$}
        \ \big\}
    \end{align*} 
    (We also refer to a tuple $G = (R,U,P,M)$ which do not actually correspond to rotations etc. as an abstract movement graph, and define $M_{\mathsf{abs}}(G)$ as above.)

    Define the \emph{(matching) movement lattice}, denoted $M(G)$, as follows.
    (As in \autoref{def:rotation-dag}, we require that for every $X =  \{ \rho_{i(1)}, \ldots \rho_{i(k)} \} \in M_{\mathsf{abs}}(G)$ where $i(1) < \ldots < i(k)$, we have $\rho_{i(j+1)}$ valid at $\rho_{i(j)}\ldots\rho_{i(1)}\mu$ for each $j=0,\ldots,k-1$, and define $\match(X) = \rho_{i(k)}\ldots\rho_{i(1)}\mu$.)
    We set
    \[ M(G) = \left\{ \match(X)\ \big|\ X \in M_{\mathsf{abs}}(G) \right\}. \qedhere \]
\end{definition}

To begin to build up a theory of movement lattices, we first show that 
that every movement lattice is in fact a lattice, using our notion of fanout lattices.

\begin{proposition}
    For every abstract movement graph $G = (\mu, R,U,P,M)$, there exists a set of fanouts $E$ over $R$ such that $M(G) = F\big((\mu, R,E)\big)$.
\end{proposition}
\begin{proof}
    Fix $G = (\mu,R,U,P,M)$, and
    for every possible $S \subseteq R$ and every PCA cycle $C$ of $S$, consider any $(\alpha, x_i)$ for $x_i \in C$ and $\alpha \in S$ as in Item~(1) of the definition of a PCA cycle.
    Moreover, for this given $C$, let $T \subseteq R$ be the set of all $\beta \in C$ such that the edge $e$ directed out of $\beta$ in $C$ is such that $e \in M$.
    Then, observe that PCA cycle $C$ corrects $S \subseteq R$ if and only if fanout $(\alpha, T)$ corrects $S$.
    Let $E$ denote the collection of fanouts which take the above form for some PCA cycle $C$ and rotations $\alpha$.
    Then, treating each predecessor edge $p = (r_1, r_0) \in P$ as a fanout $(r_1, \{r_0\})$ in the natural way, we see that
    $M(G) =  F(\mu, R, P \cup E)$.
\end{proof}

In particular, this implies that (abstract) movement lattices are in fact lattices, and they are closed under set unions, corresponding to the fact that $\PN$ is closed under coordinatewise maxima.

We now define, for any fixed profile of applicant preferences and institution priorities, a movement lattice which represents $\PN$.
This is the main construction of our paper.

To begin, consider the rotation DAG $(\IPDA, R, P_1)$ which represents of $\K\supseteq\PN$, as given by \autoref{thrm:K-has-rotations}.
The rotations of our movement graph will be $R$, and the predecessor relations will be a superset of $P_1$.
The terminal pairs $U$ are all pairs $(\EADAM(d), d)$ for $d\in \Appls$.

The predecessor relations $P\subseteq (R\cup U)\times(R\cup U)$ of the movement graph come in two different types, as for the rotation DAG of stable matching discussed in \autoref{sec:aditional-prelims-rotation-DAG}, as follows:

\begin{enumerate}
  \item Rotation $\rho_1\in R$ is a \emph{type 1 predecessor} of rotation $\rho_2\in R$ if $(\rho_2, \rho_1)\in P_1$.\footnote{
    One can show that $P_1$, i.e., all predecessor relationships of $\K$, exactly correspond to the type 1 predecessor relationships of the stable latices which represents $\K$.  
    Thus, $\K$ in fact has \emph{no} type 2 predecessor relationships.
    However, intuitively speaking, $\PN$ must (unlike $\Legal$) respect type 2 predecessor relationships.
  }

  \item Rotation $\rho_1\in R$ is a \emph{type 2 predecessor} of rotation $\rho_2\in R$ if there exists a $(d,h)\in\Appls\times\Insts$ such that $\rho_1$ moves $d$ from below $h$ to above $h$, and $\rho_2$ moves $h$ from above $d$ to below $d$.
  In other words, we have:
  \begin{align*}
       & \rho_2 = [ \ldots, (\cdot, d_{\text{down}}), (h, d_{\text{up}}), \ldots]
       && \qquad \text{if we have $h_{\text{up}} \succ_d h \succ_d h_{\text{down}}$, }
    \\ & \downarrow \text{  (type 2 predecessor) }
       && \qquad \text{and also $d_{\text{up}} \succ_h d \succ_h d_{\text{down}}$.}
    \\ & \rho_1 = [ \ldots, (h_{\text{down}}, d), (h_{\text{up}}, \cdot), \ldots]
  \end{align*}
  Additionally, terminal pair $\sigma_1 \in U$ is a \emph{type 2 predecessor} of rotation $\rho_2\in R$ if there exists a $(d,h)\in\Appls\times\Insts$ such that $d$ prefers $h$ to $d$'s paired partner in $\sigma_1$, and $\rho_2$ moves $h$ from above $d$ to below $d$.
  In other words, we have:
  \begin{align*}
       & \rho_2 = [ \ldots, (\cdot, d_{\text{down}}), (h, d_{\text{up}}), \ldots]
       && \qquad \text{if we have $ h \succ_d h_{\text{fin}}$, }
    \\ & \downarrow \text{  (type 2 predecessor) }
       && \qquad \text{and also $d_{\text{up}} \succ_h d \succ_h d_{\text{down}}$.}
    \\ & \sigma_1 = [ (h_{\text{fin}}, d) ]
  \end{align*}
\end{enumerate}

Finally, the movement relations $M \subseteq (R\cup U)\times(R\cup U)$ are defined as follows:
\begin{itemize}
  \item Rotation $\rho_1\in R$ is a \emph{movement predecessor} of rotation or terminal pair $x \in R \cup U$ if there exists a $(d,h)\in\Appls\times\Insts$ such that $\rho_1$ moves $d$ from below $h$ to above $h$, and $h$ appears in $x$.
  In other words, we have:
  \begin{align*}
       & x = [ \ldots, (h, \cdot), \ldots]
       && \qquad \text{if we have $h_{\text{up}} \succ_d h \succ_d h_{\text{down}}$. }
    \\ & \downarrow \text{  (movement predecessor) }
    \\ & \rho_1 = [ \ldots, (h_{\text{down}}, d), (h_{\text{up}}, \cdot), \ldots]
  \end{align*}

  \item Terminal pair $\sigma_1\in R$ is a \emph{movement predecessor} of rotation or terminal pair $x \in R \cup U$ if there exists a $(d,h)\in\Appls\times\Insts$ such that $d$ prefers $h$ to their paired partner in $\sigma_1$, and such that $h$ appears in $x$.
  In other words, we have:
  \begin{align*}
       & x = [ \ldots, (h, \cdot), \ldots]
       && \qquad \text{if we have $h \succ_d h_{\text{fin}}$. }
    \\ & \downarrow \text{  (movement predecessor) }
    \\ & \sigma_1 = [ (h_{\text{fin}}, d) ]
  \end{align*}
\end{itemize}

Before proceeding, we illustrate the above movement graph construction using the two examples of \autoref{sec:fanout-examples}.
\autoref{fig:implemented-fanouts} shows the constructions (omitting those terminal pairs which do not participate in any PCA cycles).

\begin{figure}[htbp]
    \begin{minipage}[c]{\textwidth}
        \centering
        \begin{minipage}[c]{0.45\textwidth}
        \begin{center}
            \includegraphics*[height=1.1\textwidth]{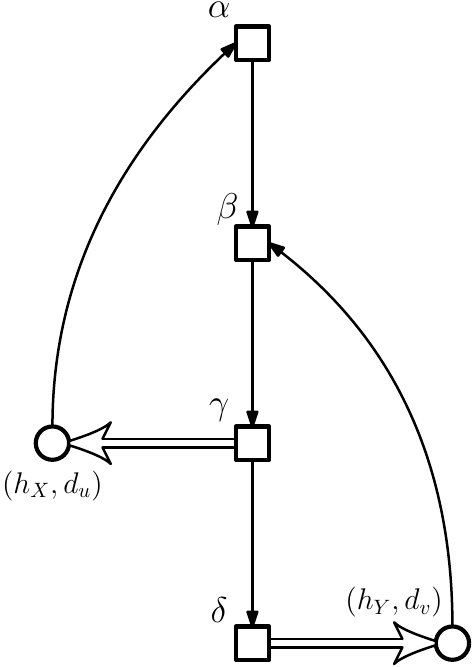}
            \subcaption{Movement graph for \autoref{fig:advanced-intersecting-fanouts}.}
        \end{center}
    \end{minipage} 
    \qquad
    \begin{minipage}[c]{0.45\textwidth}
        \begin{center}
            \includegraphics*[height=1.1\textwidth]{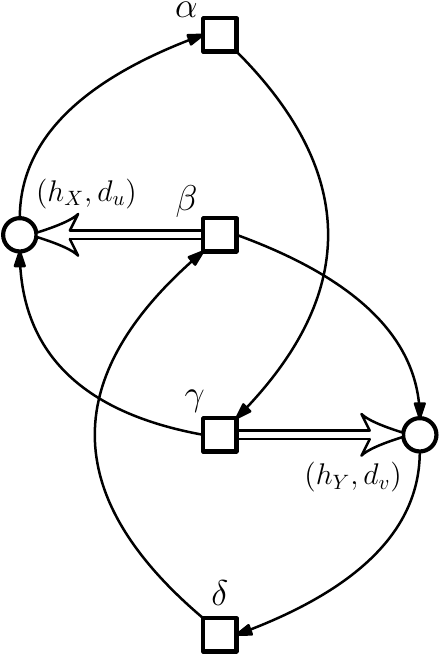}
            \subcaption{Movement graph for \autoref{fig:hossy-wide-fanouts}.}
        \end{center}
        \end{minipage} 
    \end{minipage}
    \caption{Movement graphs construction for examples in \autoref{sec:fanout-examples}.}
    \label{fig:implemented-fanouts}
    \vspace{0.1in}
    {\footnotesize
    \textbf{Notes:} 
    The doubled-arrows $x \Longrightarrow y$ represent predecessor relationships $(x,y)$; in these examples, all such relations are type 2.
    The arrows $x \longrightarrow y$ represent movement relationships $(x,y)$.
    \par}
\end{figure}

Now, let $G = (\IPDA, R, U, P, M)$ be the movement DAG as defined above.

We now show that our construction successfully represents $\PN$ with the following formal result.

\begin{proposition}
  \label{thrm:all-pn-are-movement}
  If $G$ is the movement graph defined as above, then $M(G) = \PN$, and $G$ has the good chain property.
\end{proposition}

We prove this proposition next in \autoref{sec:all-pn-are-movement-proof}, which is entirely dedicated to this proof.

\subsubsection{Proof of \autoref{thrm:all-pn-are-movement}}
\label{sec:all-pn-are-movement-proof}

Fix $G = (\IPDA, R, U, P, M)$.
To begin, observe that by the construction of $(\IPDA, R, P_1)$ representing $\K$, we have $D(G)\subseteq \K$.
The main challenge of the proof is to shown that, for all $\mu \in \K$, 
we have $\mu \in \PN$ if and only if $\mu \in D(G)$, i.e. $\rots(\mu) \in D_{\mathsf{abs}}(G)$, where we recall that $\rots(\mu)\subseteq R$ denotes the unique subset of rotations $S$ such that $\match(S)=\mu$.
Using the language of priority-correcting adjustments (\autoref{def:priority-neutral}), the proof then amounts to the following two steps.
First, we show that if $\mu \in \K$ has a priority-correcting adjustment, then $\rots(\mu)$ has a violating predecessor relationship or a correcting PCA cycle in $G$.
Second, we show that if $S\subseteq R$ with $\match(S) \in \K$ is such that $S$ has a violating predecessor relationship or a correcting PCA cycle, then $\match(S)$ has a priority-correcting adjustment.

\begin{lemma}
   \label{lem:all-pn-are-movement-real-pca-to-graph-pca}
   Suppose $\mu \in \K \setminus \PN$. %
   Then, $\rots(\mu)\subseteq R$ has either a violating predecessor relationship, or a correcting PCA cycle, in $G$.
   Hence, $D(G) \subseteq \PN$.
\end{lemma}
\begin{proof}
    Fix such a $\mu$ and let $\eta$ be any priority correcting adjustment of $\mu$.
    We consider two cases.
    First, suppose $\eta \in \K$.
    Then, there must be some $\alpha \in \rots(\mu)$ and $\beta \in \rots(\eta) \setminus \rots(\mu)$
    such that $\beta$ is a predecessor of $\alpha$.
    To see why, note that since $\eta$ is a priority correcting adjustment, there must exist some $d_*$ whose priority is violated, say at $h_*$; let $\alpha$ be the rotation which moves $h_*$ to their match in $\mu$ (such a rotation must exist because $d_*$'s priority is not violated in the base matching $\IPDA$) and let $\beta$ be the rotation which causes $d_*$ to improve when changing from $\mu$ to $\eta$.
    By the definition of $G$, we see that $\beta$ is a 
    predecessor of $\alpha$, so there must exist some violating predecessor of $\rots(\mu)$.
    This concludes the proof of the first case.

    Before proceeding, observe the following for any $\mu \in \K$.
    If there exists any $d_* \in \vio(\mu)$ such that $\mu(d_*) \ne \EADAM(d_*)$, then there must exist a rotation $\rho \in R$ which improves the match of $d_*$, and in fact $\rho\mu$ Pareto improves $\mu$. (This observation is closely related to \cite[Lemma 7]{Reny22}, which proves that if $d_*\in\vio(\mu)$ for $\mu\in\PN$ then $\mu(d_*) = \EADAM(d_*)$.) This means that, if this $d_*$ is not matched to $\EADAM(d_*)$ in $\mu$ such that $d_* \in \vio(\mu)$, then there exists some priority-correcting adjustment $\eta$ of the form handled in the previous paragraph.

    We now proceed to the second case; suppose that $\eta$ is a priority correcting adjustment of $\mu$, and that $\eta \notin \K$.
    By the previous paragraph, it is without loss of generality that we have some $d_* \in \vio(\mu)$, say $d_*$'s priority is violated at $h_*$, and such that $\mu(d_*) = \EADAM(d_*)$ and $\eta(d_*) \succ_{d_*} \mu(d_*)$.
    Now, we construct a PCA cycle $C$ in the movement graph.
    To do this, start with $(d_*, \EADAM(d_*)) \in U$.
    By our \autoref{thrm:pca-dominates-ipda}, we must have every applicant matched in $\eta$ at least as highly as in $\IPDA$ (and in particular an analogue of the rural hospitals theorem holds---every applicant matched in $\IPDA$ is matched in $\mu$ and $\eta$). 
    Thus, construct the following sequence of applicants:
    $d_1 = d_*$, and for $i \ge 1$, set $h_{i+1} = \eta(d_i)$ and $d_{i+1} = \IPDA(h_{i+1})$. 
    Since $\eta \ge \IPDA$ and this sequence cannot continue forever, there must exist a $K\ge 2$ such that $d_{K+1} = d_1 = d_*$.
    Now, we claim that given this cycle $d_1 \to \ldots \to d_K \to d_1$,
    one can construct a PCA cycle of $\rots(\mu)$ in $G$.
    (Note that there may be additional applicants who changes matches from $\IPDA$ to $\eta$ other than of those $d_i$ we construct here, but still these $d_i$ suffice to construct the PCA cycle.)

    To construct a PCA cycle $C$ of $\rots(\mu)$, start with $x_1 = (d_*, \EADAM(d_*))$. Intuitively, we then compress the chain of reassignments $C' = [ d_1\to d_2\to\ldots\to d_K\to d_1 ]$ defined as above, using edges in $P\cup M$ to change elements of $R \cup U$ whenever the sequence involves reassignments that do not correspond to any rotation.
    In more detail, given $x_i$ for $i\ge 1$, we construct an $x_{i+1}$ as follows: first, if $x_i \in R$, ignore any next-steps in $C'$ which simply correspond to reassignments that happen within $x_i$, until you reach some other reassignment. 
    (Note that this step will always necessarily halt, and this step will be trivial whenever $x_i \in U$.)
    Then, when you reach $d_j \to d_{j+1}$ in $C'$, where $d_j$ moves to some institution other than the one dictated by the movement in $x_i$, we let $x_{i+1}$ be the unique element of $R\cup U$ where $d_{j+1}$ is paired to $\mu(d_{j+1})$.

    We claim that $C$ must be a PCA cycle of $\rots(\mu)$. To see this, note that by the way we start with $x_1 = (d_*, \EADAM(d_*))$, where $d_*$'s priority is violated at $h_*$ (say, due to some rotation $\alpha \in \rots(\mu)$ which sends $h_*$ below $d_*$ in priority), we must meet Condition~(1) of being a PCA cycle from \autoref{def:movement-graph}.
    Second, since $\eta$ was a priority-correcting adjustment of $\mu$, we know that no applicant $d \in \vio(\eta)$ is such that $\eta(d)\prec_d \mu(d)$.
    Many reassignments along the chain $C'$ in fact put $d$ higher than $\mu(d)$; these can correspond to rotations $x_i \in C$ with $x_i \notin \rots(\mu)$.
    Other reassignments in fact put $d$ lower than $\mu(d)$; however, since their priority is violated \emph{nowhere} in $\eta$, it's certainly the case that they are not violated \emph{by the applicant who receives their old match} from $\mu(d)$.
    Thus, at every step along the cycle $C$, we also meet Condition~(2) of being a PCA cycle from \autoref{def:movement-graph}.
    This shows that $\rots(\mu)$ has a PCA cycle, as desired, and thus finishes the proof of \autoref{lem:all-pn-are-movement-real-pca-to-graph-pca}.
\end{proof}

We now establish the opposite inclusion.

\begin{lemma}
   \label{lem:all-pn-are-movement-graph-pca-to-real-pca}
   Suppose $S \subseteq R$ has either a violating predecessor relationship, or a correcting PCA cycle, in $G$.
   Then, $\match(S) \in \K\setminus\PN$.
   Hence, $\PN \subseteq D(G)$.
\end{lemma}
\begin{proof}
    Fix $S$ as in the statement of the lemma.
    Our goal is to show that $\match(S)$ has a priority-correcting adjustment.
    We consider two cases.
    First, suppose $S\subseteq R$ is such that some predecessor relation in $P$ violates $S$.
    (Similarly to the proof of \autoref{lem:all-pn-are-movement-real-pca-to-graph-pca},
    this serves as the ``easy case''.)
    If some element of $P_1$ in fact violates $S$, then we do not even have $\match(S) \in \K$. (In fact, one can show that $\match(S)$ is not even well-defined, since there will be no ordering of the rotations in $S$ that can be successively eliminated from $\IPDA$.)
    If some element of $P_2$ in fact violates $S$, then by construction we know there is some rotation $\beta\notin S$, and some applicant $d_*$ whose priority is violated in $\match(S)$, such that $\beta$ strictly improves the match of $d_*$ that of $\match(S)$. This shows $\match(S)$ has a priority correcting adjustment, and thus that $\match(S)\notin\PN$.

    Second, consider the harder case in which $S\subseteq R$ has a correcting PCA cycle $C$ in $G$. 
    Call this cycle $C = x_1 \to x_2 \to \ldots \to x_K \to x_1$ with corresponding edges $e_1 = (x_1,x_2),\ldots, e_K = (x_K,x_1)$.
    By Condition~(1) of the definition of a PCA cycle, there exists an element of $C$, say $x_1$, such that
    $x_1\in U$ and there is an $\alpha \in S$ such that $x_1$ is a predecessor of $\alpha$.
    Additionally, by Condition~(2), for all $e_i = (x_i, x_{i+1})$ in $C$, if $e_i$ is a movement edge, then then $x_i \notin S$.
    Now, one could define a matching by starting from $\match(S)$, and iteratively reassigning agents exactly according to the edges $e_1,\ldots,e_K$ and the corresponding rotation $x_i \in R$ where needed (i.e., via the reverse of the process used to reassign according to rotations in \autoref{lem:all-pn-are-movement-real-pca-to-graph-pca});
    this would indeed improve the match of $d_*$ by construction.
    Moreover, since $x_i \in S$ can appear in $C$ only if $e_{i} \in P$,
    we know that each reassigned $d_j$ cannot have their priority at their old match $\mu(d_j)$ violated.
    Let this new matching be denoted $\eta$.
    However, interestingly, for arbitrary PCA cycles, this might not actually be a priority correcting adjustment of $\match(S)$.
    The reason is that, for those applicants $d_j$ whose match we \emph{worsen} in $\eta$ relative to $\mu$, there might exist $\overline{h}$ other than $\mu(d_j)$ where $d_j$'s priority is violated.

    The key to avoiding the above problem is the following observation.
    For such a $d_j$ and $\overline{h}$, there also exists a movement edge from 
    the element of $R\cup U$ where $d_j$ was assigned to the element of $R\cup U$ where $\overline{h}$ is reassigned according to the process constructing $\eta$.
    Moreover, if the movement edge assigning $\eta(\overline h)$ to $\overline h$ was an interrupter edge, then since (by assumption) $d_j$ has higher priority than $\eta(\overline h)$ at $\overline h$, we know that the edge from $d_j$ is also a interrupter edge.
    Thus, reassigning $d_j$ to $\overline{h}$, and continuing reassigning according to the remainder of $C$ after $\overline{h}$, is also a valid PCA cycle $C'$.
    Moreover, observe that this PCA cycle is strictly smaller in the set-inclusion sense than the original $C$.
    Since this is the only problem that can arise when constructing $\eta$, and since we can always eliminate the problem while strictly decreasing the set of element of $R\cup U$ in $C$, repeatedly doing this reduction must eventually terminate in a scenario where this problem does not exist.
    At this point---in other words, at any set-inclusion-wise minimal $C$---constructing such at $\eta$ as above must thus create a priority correcting adjustment of $\match(S)$.
    This finishes the proof of \autoref{lem:all-pn-are-movement-graph-pca-to-real-pca}.
\end{proof}

Since \autoref{lem:all-pn-are-movement-graph-pca-to-real-pca} shows that if $S\subseteq R$ has a PCA cycle, then $\match(S)\notin\PN$, and our good chain result for $\PN$ (\autoref{lem:struct-for-pn-max-chain}) shows that there exists a maximal chain in $\K$ whose rotations sets are form $S_1,\ldots,S_N$ with $S_i = \{\alpha_1,\ldots,\alpha_i\}$, and such that $\match(S_i) \in \PN$, we thus observe that $G$ must have the good chain property.
\begin{observation}
    $G$ has the good chain property.
\end{observation}

This shows that every $\PN$ lattice can be represented via a movement graph with the good chain property, and thus concludes the proof of \autoref{thrm:all-pn-are-movement}.

\subsection{All Movement Lattice are \texorpdfstring{$\PN$}{PN} Lattices}
\label{sec:all-movement-is-PN}

We now establish a converse of \autoref{thrm:all-pn-are-movement}: that all movement lattices with the good chain property correspond to some $\PN$ lattice.
To do this, we construct for each such movement graph an explicit instance of priorities and preferences with the corresponding $\PN$ lattice.
Our construction generalizes a method of constructing a stable lattice isomorphic to any given distributive lattice; for completeness and exposition, we provide this construction (which we believe is simpler than previous ones from \cite{blair1984every, GusfieldI89}) in \autoref{sec:exposition-stable-construction}.

\begin{proposition}
  \label{thrm:all-movement-are-pn}
  Consider any movement graph $G$ with the good chain property.
  Then, there exists a profile of preferences and priorities 
  such that $\PN = M(G)$.
\end{proposition}

The remainder of this subsection is dedicated to the proof of \autoref{thrm:all-movement-are-pn}.

Fix some abstract movement graph $G = (R, U, P, M)$, say with ordering $[\alpha_1, \ldots, \alpha_N]$ over $R$ such that the good chain property holds, which induces abstract movement lattice $M(G)$.
We will construct a instance of preferences and priorities such that $\PN$ is isomorphic to  $M(G)$.
For each $i\in [N]$, this construction has applicants $d_i, d_i'$ and institutions $h_i, h_i'$;
the preferences will be such that for each $i$, we have $\IPDA(d_i) = h_i$ and $\IPDA(d_i')=h_i'$, as well as $\EADAM(d_i) = h_i'$ and $\EADAM(d_i')=h_i$. 
Moreover, we construct for each $u\in U$ an applicant $d_u$ and institution $h_u$ such that $\IPDA(d_u) = \EADAM(d_u) = h_u$. 
Preferences and priorities are decided by the predecessor edges $P$ and movement edges $M$ as follows.

For each $i\in [K]$, we set the priorities of $h_i$ and $h_i'$ as follows:
\begin{align*}
    h_i & : d_i \succ \mathcal{S}_i \succ d_i' \succ \ldots \\
    h_i' & : d_i' \succ d_i \succ \ldots
\end{align*}
where $\mathcal{S}_i$ is the set of all $d_j$ such that in $G$, there is an predecessor edge from $\alpha_i \in R$ to $\alpha_j \in R$,
along similarly with the set of all and $d_u$ such that in $G$, there is an predecessor edge from $\alpha_i \in R$ to $u\in U$.
We set the preferences of $d_i$ and $d_i'$ as follows:
\begin{align*}
    d_i & : h_i' \succ \mathcal{T}_i \succ h_i \succ \emptyset \\
    d_i' & : h_i \succ h_i' \succ \emptyset
\end{align*}
where $\mathcal{T}_i$ is the set of all $h_j$ such that in $G$, there is a predecessor \emph{or} movement edge from $\alpha_j\in R$ to $\alpha_i \in R$, 
as well as all $h_u$ for $u \in U$ such that there is a movement edge from $u$ to $\alpha_i$.

Next, we set for each $u \in U$ the priorities of $h_u$ to be $d_u \succ \ldots$, i.e., any priorities that rank $d_u$ first.
We set the preferences of $d_u$ for each $u \in U$ as follows:
\begin{align*}
    d_u & : \mathcal{V}_u \succ h_u \succ \emptyset
\end{align*}
where $\mathcal{V}_u$ is the set of all $h_j$ such that there is a movement \emph{or} predecessor edge from $\alpha_j$ to $u$ in $G$, along with all $h_v$ such that there is a movement edge from $v\in U$ to $u\in U$ in $G$.
This completes the definition of the instance of preferences and priorities we consider.

Now, the main challenge for the remainder of the proof will be to show that under these preferences and priorities, the set of rotation in the predecessor DAG representation of $\K$ is exactly the set of all $\beta_j := [ (h_j, d_j), (h_j',d_j') ]$ for $j\in[K]$.
For, if we can show this, then one can check that directly applying the construction above \autoref{thrm:all-pn-are-movement} will give the desired movement graph representation (except for extra terminal pairs of the form $(d_j,h_j')$ and $(d_j',h_j)$, which will not participate in any PCA cycles and hence not effect $M(G)$).
To establish this set of rotations, we will crucially use the good chain property, along with arguments from \cite{FaenzaZ22} that allow us to find the set of rotations in the lattice $\Legal$.

The key claim is thus as follows:
\begin{lemma}
    \label{lem:all-movement-are-pn-correct-rotations}
    With preferences and priorities as above, consider the predecessor DAG representation $(R', P')$ of $\K$ (as defined in \autoref{thrm:K-has-rotations}).
    Then, we have $R' = \{ \beta_j \}_{j\in[K]}$, where we write $\beta_j = [ (h_j, d_j), (h_j',d_j') ]$ for each $j$. 
\end{lemma}
\begin{proof}
    Recall that we assumed that in the original movement graph, the ordering $[\alpha_1,\ldots,\alpha_N]$ on $R$ was assumed to have the good chain property.
    We prove this claim using Algorithm~1 of \cite{FaenzaZ22}, school-rotate-remove.
    When run on the above profile of preference and priorities,
    observe that after the $j$th step of the main loop of their algorithm, we will inductively have that the current assignment is something of the form $\match(\{\beta_1, \beta_2, \ldots, \beta_k\})$.
    After this, it is not hard to see that our good chain property implies that the iteration of the $\APDA$ algorithm on the trimmed instance consider in \cite{FaenzaZ22} (where that run of $\APDA$ is implied by their definition of the school-rotation digraph), we will always proceed to a 
    further assignment of the form $\match(\{\beta_1, \beta_2, \ldots, \beta_{k+x}\})$ for some $x \ge 1$.
    Hence, Algorithm~1 of \cite{FaenzaZ22} removes all sets $\mathcal{S}_i$ from the priority list of each $h_i$, and all $\mathcal{T}_i$ from the preference list of each $d_i$.  
    Once this is done, 
    the legal lattice then becomes the stable lattice on these set of trimmed preferences, 
    and these stable rotations equal $\{ \beta_j \}_{j \in [K]}$, as needed.
\end{proof}

\autoref{lem:all-movement-are-pn-correct-rotations} thus establishes, along with \autoref{thrm:all-pn-are-movement}, that $\PN = M(G)$.
This concludes the proof of \autoref{thrm:all-movement-are-pn}.

\subsection{Conclusion and Discussion}
\label{sec:discussion}

We now sum up this section by stating our main result: movement lattices with the good chain property exactly characterize priority-neutral lattices.
This follows directly from \autoref{thrm:all-pn-are-movement} and \autoref{thrm:all-movement-are-pn}.
Formally:
\begin{theorem}[Main Theorem; Characterization of $\PN$ lattices]
    \label{thrm:pn-iff-movement}
    A lattice arises as $\PN$ for some profile of preferences and priorities if and only if it is isomorphic to a movement lattice with the good chain property.
\end{theorem}

We now briefly conclude by mentioning two implications of our characterization.

First, to illustrate why $\PN$ lattices are a strict subset of the class of lattices identified in \autoref{sec:simpler-sufficient},
we show that the lattice in \autoref{fig:n5}, which we recall is known as $N_5$, cannot arise as a $\PN$ lattice (even though it can be represented by a fanout set with the good chain property).

\begin{proposition}
    \label{thrm:n5-not-priority-neutral}
    The lattice $N_5$, shown in \autoref{fig:n5-lattice}, cannot arise as the $\PN$ lattice.
\end{proposition}
\begin{proof}
  Suppose for contradiction that $\PN$ is isomorphic to $N_5$.
  To begin, we show that if the lattice $\PN$ is isomorphic to \autoref{fig:n5-lattice}, then (up to relabeling) the corresponding fanout multigraph is as in \autoref{fig:n5-fanout}.
  To see this, first we observe that the good chain of $\PN$ in \autoref{fig:n5-lattice} must be the one labeled $\emptyset, \alpha, \alpha\beta, \alpha\beta\gamma$ in the figure. For, otherwise, the node labeled $\beta\gamma$ in the figure would have to correspond to a single rotation $\delta$ which is not in either of the nodes in the left of the figure, and thus we could not have the join of $\delta$ with the nodes on the left both equal to $\alpha\beta\gamma$.
  Second, this means no fanout can come out of rotation $\alpha$; a fanout must come out of $\beta$ to $\alpha$, which one can check must also include $\gamma$; and one can check that a fanout must come out of $\gamma$ to $\beta$.

  Now we examine two cased based on the fanout which comes out of $\beta$, and use the movement graph representation of $\PN$.
  In both cases, some terminal pair $p$ must be a type 2 predecessor of $\beta$.
  In the first case, suppose the PCA cycle involving $p$ which defines the fanout $(\beta, \{\alpha,\gamma\})$ first goes to $\alpha$.
  Then, one can check that no matter how the fanout $(\gamma, \{\beta\})$ is defined, there will be another fanout $(\beta, \{\alpha\})$.
  Second, suppose the above-mentioned PCA cycle first goes to $\gamma$. 
  Then, one can similarly check that no matter how the fanout $(\gamma, \{\beta\})$ is defined, there will be another fanout $(\beta, \{\gamma\})$.
  Both of these are contradictions.
  This concludes the proof.
\end{proof}

Thus, the class of $\PN$ lattices is a subtle, and we believe quite novel, family of lattices.

Second, we note that since our characterization is constructive, it in fact gives a polynomial time algorithm for checking whether a matching in $\PN$.
This algorithm constructs the movement graph, and then, on input $\mu \in \K$, checks whether $\mu \in \PN$ by considering the set of rotations $S$ corresponding to $\mu$, removing all edges that cannot participate in a PCA cycle for $S$, considering every $u \in U$ where $(\alpha, u)\in P$ for some $\alpha \in S$, and checking whether a cycle includes $u$.\footnote{
    Note that the movement relationships of the construction can create up to $n$ different movement edges in $M$ for each $(d,h)$ pair---one for each of the up to $n$ rotation where $h$ appears---and thus causes our construction overall to use $\widetilde O(n^3)$ bits rather than the tight $\widetilde O(n^2)$.
    There is an easy way to circumvent this, though---one can add to $G$ ``institution nodes'' for each $h$, and represent each movement edge in $M$ via an edge from an element of $R\cup U$ to such an $h$, and from such an $h$ to some element of $R\cup U$.
    This would result in a $\widetilde O(n^2)$ bit data structure, and the arguments of \autoref{lem:all-pn-are-movement-graph-pca-to-real-pca} show that the definition of a PCA cycle can be adapted to this modified graph.

    Even with the above optimization, however, the runtime of the implied algorithm for checking priority neutrality is at least $\Omega(n^4)$, since we check separately for each rotation whether there is a corresponding PCA cycle, rather than the tight $\widetilde O(n^2)$.
    A tight algorithm remains an interesting direction for future work.
}
Given only the definition of $\PN$, it is far from clear that such an algorithm should exist, i.e., why checking for membership in $\PN$ should be any easier than check, for all possible matchings $\xi$ separately, whether $\xi$ is a priority correcting adjustment.
This emphasizes that, while priority-neutral matching lattices are significantly more involved than stable ones, priority-neutral matchings have some degree of tractable structure.
Developing further algorithms for $\PN$ remains an interesting direction for future work which, like for stable matching, likely requires exploiting the structure of $\PN$ in detail.

\bibliographystyle{alpha}
\bibliography{Bib}{}

\appendix

\clearpage

\section{Simple Construction of a Stable Matching Lattice for Any Distributive Lattice}
\label{sec:exposition-stable-construction}

In this brief appendix, we re-prove the known result that any abstract rotation DAG (\autoref{def:rotation-dag}) can be realized as the stable matching lattice.
Our construction simplifies prior ones \cite{blair1984every,GusfieldI89}, and provides exposition to our construction of movement lattices as priority-neutral matching lattices in \autoref{sec:all-movement-is-PN}.

Let $G$ be an abstract rotation DAG with rotations $R = \{1,\ldots,n\}$ and predecessor relations $P$.
For each $i \in R$, let $S_i$ be the set of all $d_j$ such that $(i,j)\in P$, and let $T_i$ be the set of all $h_jj$ such that $(j,i)\in P$.
Then, let the preferences and priorities be as in \autoref{fig:stable-construction}, where for each such $S_i$ and $T_i$, preferences and priorities may be placed in any order.

\begin{figure}[htb]
\begin{minipage}{\textwidth}
\centering
    For $i\in R = \{1,\ldots,n\}$:
    \\[0.05in]
    \begin{tabular}{ccccccc}
        \toprule
        $h_i$ & $h_i'$ \\
        \midrule
        $d_i$ & $d_i'$ \\
        $S_i$ & $d_i$ \\
        $d_i'$ &
    \end{tabular}
    \qquad
    \begin{tabular}{ccccccc}
        \toprule
        $d_i$ & $d_i'$ \\
        \midrule
        $h_i'$ & $h_i$ \\
        $T_i$  & $h_i'$ \\
        $h_i$ &
    \end{tabular}
    \caption{Preferences and Priorities.}
    \label{fig:stable-construction}
\end{minipage} 
\end{figure}

This simple construction suffices:
\begin{proposition}
  With preferences and priorities as in \autoref{fig:stable-construction}, we have that $\Stab$ is isomorphic to $D_{\mathsf{abs}}(G)$.
\end{proposition}
\begin{proof}
  To begin, observe that $\IPDA$ matches $h_i$ to $d_i$ and $h_i'$ to $d_i'$ for each $i$, and $\APDA$ matches
  $h_i$ to $d_i'$ and $h_i'$ to $d_i$ for each $i$.
  Moreover, since the rotation DAG $G$ is acyclic, there is no feasible matching $\mu$ with $\IPDA < \mu < \APDA$ such that any $d_i$ is matched to $h_j$ for some $j\ne i$.
  Additionally, if we let $\alpha_i$ denote the rotation $[(h_i, d_i), (h_i', d_i')]$ for each $i$, acyclicity implies that the matching $\alpha_k\alpha_{k-1}\ldots\alpha_2\alpha_1\IPDA$ is stable for each $1 \le k \le n$, since there is no $i > k$ where $T_i$ contains $h_j$ for some $j \le k$.
  Hence, the rotations of $\Stab$ are exactly $\{\alpha_1,\ldots,\alpha_n\}$.
  Moreover, it is direct from the definitions that the predecessor relations in $\Stab$ are all type 2, and are exactly the predecessor edges in $G$.
  Hence, if we let $P'$ be the predecessor edges on $\{\alpha_1,\ldots,\alpha_n\}$ which exactly correspond to $P$, then $\Stab = D(\IPDA, \{\alpha_1,\ldots,\alpha_n\}, P')$ is isomorphic to $D_{\mathsf{abs}}(G)$, which finishes the proof.
\end{proof}

Since Birkhoff's representation theorem \cite{Birkhoff37} implies that any finite distributive lattice has a representation via a rotation DAG (namely, with any rotation DAG that gives the correct partial order on the join-irreducible elements of the lattice), this directly shows that all distributive lattices arise as stable matching lattices.

In our opinion, our construction is simpler than that of \cite{GusfieldI89}, which creates a separate applicant and institution for each immediate predecessor relation in the rotation DAG.
(The approach of \cite{GusfieldI89} is, in turn, substantially more streamlined approach than that of \cite{blair1984every}, which may use exponentially more applicants and institutions than is necessary to construct a given lattice.)
However, one can show that both our construction and that of \cite{GusfieldI89} may use quadratically more applicants and institutions than is necessary to construct a given lattice (intuitively, since the most-efficient construction may re-use a given applicant or institution linearly many times to create different rotations).
Determining the smallest instance such that the stable lattice is a given distributive lattice is likely a very intractable problem in general.

\end{document}